\documentclass[a4paper,11pt]{article}
\usepackage{bm,mathrsfs,amsmath,amsthm,amssymb,eepic,amscd,longtable,array,graphicx}
\newcommand{\nc}{\newcommand}
\nc{\rnc}{\renewcommand}
\nc{\nn}{\nonumber}
\nc{\der}{{\partial}}
\rnc{\Im}{{\rm{Im}\,}}
\rnc{\Re}{{\rm{Re}\,}}
\nc{\db}{\displaybreak[0]\\}
\nc{\bra}{\langle}
\nc{\ket}{\rangle}
\nc{\bs}{\boldsymbol}

\DeclareMathOperator{\Tr}{Tr}

\DeclareMathOperator{\End}{End}

\newtheorem{theorem}{Theorem}[section]
\newtheorem{lemma}[theorem]{Lemma}

\newtheorem{proposition}[theorem]{Proposition}
\newtheorem{corollary}[theorem]{Corollary}

\theoremstyle{definition}
\newtheorem{definition}[theorem]{Definition}
\newtheorem{example}[theorem]{Example}

\numberwithin{equation}{section}

\numberwithin{equation}{section}

\begin{document}%
%
\title{Vertex models, TASEP and Grothendieck Polynomials}

\author{
Kohei Motegi$^1$\thanks{E-mail: motegi@gokutan.c.u-tokyo.ac.jp} \,
and
Kazumitsu Sakai$^2$\thanks{E-mail: sakai@gokutan.c.u-tokyo.ac.jp}
\\\\
$^1${\it Okayama Institute for Quantum Physics, }\\
 {\it Kyoyama 1-9-1, Okayama 700-0015, Japan} \\
\\
$^2${\it Institute of physics, University of Tokyo,} \\ 
{\it Komaba 3-8-1, Meguro-ku, Tokyo 153-8902, Japan}
\\\\
\\
}

\date{May 14, 2013}

\maketitle

\begin{abstract}
We examine the wavefunctions and their scalar products of
a one-parameter family of integrable five vertex models. 
At a special point of the parameter, the model investigated 
is related to an irreversible interacting stochastic particle 
system the so-called totally asymmetric simple exclusion 
process (TASEP). By combining the quantum inverse scattering 
method with a matrix product representation of the wavefunctions, 
the on/off-shell wavefunctions of the five vertex models 
are represented as a certain determinant form.  Up to some normalization 
factors, we find the wavefunctions are given by  Grothendieck polynomials,
which are a one-parameter deformation of Schur polynomials.
Introducing a dual version of the Grothendieck polynomials,
and utilizing the determinant representation for the scalar products
of the wavefunctions, we derive a generalized Cauchy identity
satisfied by the Grothendieck polynomials and their duals. Several
representation theoretical formulae for Grothendieck polynomials are also
presented. As a byproduct, the relaxation dynamics such as Green 
functions for the periodic TASEP are found to be described in 
terms of Grothendieck polynomials.
\end{abstract}


\section{Introduction}Symmetric polynomials \cite{Mac} are ubiquitous objects in 
mathematics and mathematical physics.
One of the most basic and important symmetric polynomials
is the  Schur polynomials
\begin{align}
s_\lambda(\bs{z})=\frac{\mathrm{det}_N(z_j^{\lambda_k+N-k})}{\prod_{1 \le j < k \le N}(z_j-z_k)},
\label{schur}
\end{align}
where $\bs{z}=\{z_1,\dots,z_N\}$ is a set of variables and
$\lambda=(\lambda_1,\lambda_2,\dots,\lambda_N)$
is a sequence of weakly decreasing
nonnegative integers $\lambda_1 \ge \lambda_2 \ge \dots \ge \lambda_N \ge 0$. 
A sequence $\lambda$ can be represented as a Young diagram
whose $k$th row has $\lambda_k$ boxes.

Schur polynomials appear not only in representation theory
but also in many contexts in mathematical physics,
especially in integrable systems.
For example, the tau functions of the KP hierarchy have
Schur polynomial expansions \cite{DJKM}.
Schur polynomials also appear as
the singular vectors in conformal field theory \cite{Segal,WY},
the Green function of the vicious walkers \cite{GOV,NKT},
the domain wall boundary partition function of the six vertex model \cite{Ok,St},
the Schur processes as one of the most fundamental examples
of determinantal processes \cite{OR}, to list a few.
The relation between integrable models and symmetric polynomials 
can be extended from Schur polynomials
to Jack, Hall-Littlewood, Macdonald polynomials and so on.

In this paper, we develop a novel relation between
integrable models and symmetric polynomials.
We consider a family of integrable five vertex models.
At a special point of the parameter, the model we investigate
is related to an irreversible interacting stochastic particle 
system called the totally asymmetric simple exclusion 
process (TASEP) \cite{De,Sch,Li,Sp,GM,Bo}.
At another point, the vertex models reduce to the four vertex model
describing the one-dimensional quantum Ising model \cite{Bo2}.

We show that up to normalization factors, the wavefunctions of the
five vertex model are given by the Grothendieck polynomials \cite{LS,IN,IS}
\begin{align}
G_\lambda(\bs{z};\beta)=\frac{\mathrm{det}_N(z_j^{\lambda_k+N-k}(1+\beta z_j)^{k-1})}{\prod_{1 \le j < k \le N}(z_j-z_k)}, \label{introductionGR}
\end{align}
which is a one-parameter extension of the Schur polynomials.
The Grothendieck polynomial was originally introduced \cite{LS}
in the context of the intersection between geometry
and representation theory as a $K$-theoretical
extension of the Schubert polynomials, i.e.
as polynomial representatives of Schubert classes in the
Grothendieck ring of the flag manifold. The formal parameter $\beta$
corresponds to the $K$-theoretical extension.
For flag varieties of type $A$,
Schubert polynomials is the Schur polynomials itself,
and it was shown recently \cite{IN,IS} that Grothendieck polynomials
for flag varieties of type $A$
is expressed in the determinant form \eqref{introductionGR}.
We show the equivalence between the wavefunctions
and Grothendieck polynomials
by combining the quantum inverse scattering 
method with a matrix product representation of the wavefunctions.
We find the wavefunctions correspond to the Grothendieck polynomials 
\eqref{introductionGR}, and the  dual wavefunctions
to the following  ``dual" Grothendieck polynomials
\begin{align}
\overline{G}_\lambda(\bs{z};\beta)=
\frac{\mathrm{det}_N(z_j^{\lambda_k+N-k}(1+\beta z_j^{-1})^{1-k})}
       {\prod_{1 \le j < k \le N}(z_j-z_k)}
\label{introductiondualGR}.
\end{align}

From this relation between integrable models and symmetric polynomials,
we note that studying integrable five vertex models lead us to
representation theoretical results of the Grothendieck polynomials.
We present several results for the Grothendieck polynomials
in this way, i.e. by studying the integrable five vertex models.
We find the following Cauchy identity for the
Grothendieck and dual Grothendieck polynomials
\begin{align}
&\sum_{\lambda \subseteq
(M-N)^N}G_\lambda(\bs{z};\beta)
\overline{G}_\lambda(\bs{y};\beta)
\nonumber \\
&\qquad =\prod_{1 \le j < k \le N}\frac{1}{(z_j-z_k)(y_j-y_k)}
\mathrm{det}_N \left[ \frac{(z_j y_k)^{M}-
\left\{(1+\beta z_j)/(1+\beta y_k^{-1})\right\}^{N-1}}
{z_j y_k-1} \right],
\label{introductiongeneralizedcauchy}
\end{align}
which generalizes the one for Schur polynomials.
We show this identity
by evaluating the scalar products
of the wavefunctions in two ways.
We can evaluate the determinant representation
of the scalar products directly in a way by use of the recursive relations
called the Izergin-Korepin approach.
The scalar products can also be evaluated
by inserting the completeness relation and 
the determinant forms of the wavefunctions and dual wavefunctions. 
The two ways of the evaluation of the scalar products
turns out to yield the Cauchy identity for Grothendieck polynomials 
\eqref{introductiongeneralizedcauchy}.
In short, we find a generalization of the celebrated Cauchy identity
by analyzing a family of integrable five vertex models
with the recently developed techniques to analyze integrable models 
(see \cite{Ko,Iz} for the six vertex model,
\cite{KBI,Sl,KMST,Wh,KM} for the XXZ chain, and
\cite{GMmat} for the totally asymmetric simple exclusion process).
As a special case of the Cauchy identity,
we also obtain the summation formula for the Grothendieck polynomials.
There are several results on the generalizations
of the Cauchy identity for symmetric polynomials related
to geometry in the past (see \cite{KM,Ki} for example).
The one for the Grothendieck polynomials presented
in this paper has advantages in that the connection with
the Schur polynomials is explicit and easily understandable,
which seems not to be known before.

As a byproduct of the determinant forms of the wavefunctions,
we formulate the exact relaxation dynamics of the periodic TASEP
for arbitrary initial condition,
generalizing the case for the step and alternating initial conditions
\cite{MSS,MSS2}.

This paper is organized as follows.
In the next section, we introduce a one-parameter family
of integrable five vertex models by solving the $RLL$-relation,
a version of the Yang-Baxter relation which guarantees the
integrability of the model.
In section 3, we evaluate the scalar products by the
Izergin-Korepin approach to find its determinant form.
In section 4,
by combining the quantum inverse scattering 
method with a matrix product representation of the wavefunctions,
the determinant representation of the wavefunctions is obtained.
In section 5, we establish the relation between the wavefunctions
and the Grothendieck polynomials.
Combining the results in section 3 and 4,
we derive the Cauchy identity and the summation formula
for Grothendieck polynomials.
We give a formulation of the exact relaxation dynamics
of the periodic TASEP for arbitrary initial condition in section 6.
Section 7 is devoted to the conclusion of this paper.
%
\section{One-parameter family of five vertex models}\label{FV}
%
A key ingredient in constructing quantum integrable models is to find a 
solution of the relation ($RLL$-relation)
\begin{align}
R_{\mu \nu}(u,v)L_{\mu j}(u)L_{\nu j}(v)=
L_{\nu j}(v)L_{\mu j}(u)R_{\mu \nu}(u,v)
\label{RLL}
\end{align}
holding in $\End(W_\mu \otimes W_\nu \otimes V_j)$ for arbitrary
$u,v \in \mathbb{C}$.
Here the matrix $R_{\mu \nu}(u,v)\in \End(W_\mu \otimes W_\nu)$ 
satisfies the Yang-Baxter equation
\begin{align}
R_{\mu \nu}(u,v)R_{\mu  \gamma}(u,w)R_{\nu \gamma}(v,w)=
R_{\nu \gamma}(v,w)R_{\mu  \gamma}(u,w)R_{\mu \nu}(u,v)
\label{YBE}
\end{align}
and $L_{\mu j}(u)$ is an operator acting on the space  $W_\mu \otimes V_j$. 
By convention we call $W$ and $V$ the auxiliary space and the quantum space, 
respectively.

In the following, we shall take both of the spaces $W$ and $V$  to be the 
two-dimensional vector spaces $W=V=\mathbb{C}^2$ spanned by the 
``empty state"  $|0\ket=\binom{1}{0}$ and the 
``particle occupied state" $|1\ket=\binom{0}{1}$. 
(Note that $W_\mu$ (resp. $V_j$) denotes a copy of $\mathbb{C}^2$ 
spanned by the $\mu$th (resp. $j$th) states $|0\ket_\mu$ and 
$|1\ket_\mu$ (resp. $|0\ket_j$ and $|1\ket_j$)). One solution to the 
Yang-Baxter equation \eqref{YBE} is the following $R$-matrix 
whose elements are the Boltzmann weights associated with a five
vertex model:
\begin{align}
R(u,v)
=
\begin{pmatrix}
f(v,u) & 0 & 0 & 0 \\
0 & 0 & g(v,u) & 0 \\
0 & g(v,u) & 1 & 0 \\
0 & 0 & 0 & f(v,u)
\end{pmatrix}
\label{Rmatrix}
\end{align}
with
\begin{align}
f(v,u)=\frac{u^2}{u^2-v^2},  \quad 
g(v,u)=\frac{uv}{u^2-v^2}. 
\end{align}
As a solution of the $RLL$-relation \eqref{RLL} with the $R$-matrix 
\eqref{Rmatrix}, we find the following $L$-operator 
$L(u)\in\End(\mathbb{C}^2\otimes\mathbb{C}^2)$ 
(see Appendix A for the detailed derivation):
\begin{align}
L_{\mu j}(u)=u s_\mu s_j+\sigma_\mu^- \sigma_j^+
                   +\sigma_\mu^+ \sigma_j^-+(\alpha u-u^{-1})n_\mu s_j
                   +\alpha u n_\mu n_j,
\label{loperator}
\end{align}
where $\sigma^{\pm}, \sigma^{z}$ are the spin-1/2 Pauli matrices; 
$s=(1+\sigma^z)/2$ and $n=(1-\sigma^z)/2$ are the projection 
operators onto the states $|0\ket$ and $|1\ket$, respectively.
Note that the operators with subscript $\mu$ (resp. $j$)  act on the auxiliary 
(resp. quantum) space $W_\mu$ (resp. $V_j$). 
See also Figure \ref{weight} for a pictorial description of the $L$-operator 
\eqref{loperator}, which allows for an intuitive understanding of the
subsequent calculations.

The parameter $\alpha$ can be taken arbitrary\footnote{ 
Note that the parameter $\alpha$ is different from the parameter $q$ of
 a quantum group $U_q(sl_2)$.}.
In fact, the models at special points of $\alpha$
are related to some physically interesting models. To see
this, let us construct the monodromy matrix $T(u)$
by a product of $L$-operators:
\begin{align}
T_{\mu}(u)=\prod_{i=1}^M L_{\mu j}(u)
\label{monodromy1}
\end{align}
which acts on $W_\mu \otimes (V_1\otimes\dots\otimes V_M)$. Tracing
out the auxiliary space, one defines the transfer matrix 
$\tau(u)\in \End (V^{\otimes M})$:
\begin{align}
\tau (u)=\Tr_{W_{\mu}} T_{\mu}(u). 
\label{TM}
\end{align}
Thanks to the $RLL$-relation, the transfer matrix $\tau(u)$ mutually 
commutes, i.e.
\begin{align}
[\tau(u),\tau(v)]=0.
\label{commutative}
\end{align}
After taking the logarithmic derivative of the transfer matrix with respect 
to the spectral parameter, one obtains the quantum Hamiltonian
which is, in general, non-Hermitian
\begin{align}
\mathcal{H}:=
\sum_{j=1}^M 
\left\{
\alpha \sigma_j^+\sigma_{j+1}^-+\frac{1}{4}
(\sigma_j^z\sigma_{j+1}^z-1)\right\}
=
\left.
\frac{1}{2\sqrt{\alpha}}\frac{\der}{\der u} 
\log\left\{ \left(\sqrt{\alpha} u\right)^{-M} \tau(u)\right\}
\right|_{u=\frac{1}{\sqrt{\alpha}}}.
\label{Baxter}
\end{align}
At $\alpha=1$, the $L$-operator is essentially the same with the $R$-matrix
\eqref{Rmatrix}. In this case, the quantum Hamiltonian 
$\mathcal{H}$ \eqref{Baxter} can be interpreted as 
a stochastic matrix describing an irreversible interacting stochastic 
particle system called the totally asymmetric simple exclusion process 
(TASEP) \cite{Bo,GMmat} (see section~\ref{TASEP}, for details).
On the other hand, in the limit $\alpha\to \infty$,
this model is related to  an irreversible
noncollisional diffusion process (i.e. a
vicious random walker model): 
$(\mathcal{H}/\alpha-1)^{\otimes N}|_{\alpha\to\infty}$
is nothing but a transition matrix describing the process of 
$N$ vicious random walkers.
Finally, the $L$-operator at  $\alpha=0$  
reduces to the four vertex model \cite{Bo2}, and through the
relation \eqref{Baxter} it is
related to the well-known Ising model. 

The quantum integrability of the model \eqref{Baxter} is
easily understood by the commutativity of the transfer 
matrix \eqref{commutative} and its Hamiltonian limit \eqref{Baxter}:
the transfer matrix $\tau(u)$ is just a generator of  nontrivial
conserved quantities. 

In the next section, by means of the quantum inverse scattering
method, we construct  state vectors of the Hamiltonian 
\eqref{Baxter} (or equivalently, of the transfer matrix \eqref{TM}),
and calculate their scalar products.

\begin{figure}[tt]
\begin{center}
\includegraphics[width=0.75\textwidth]{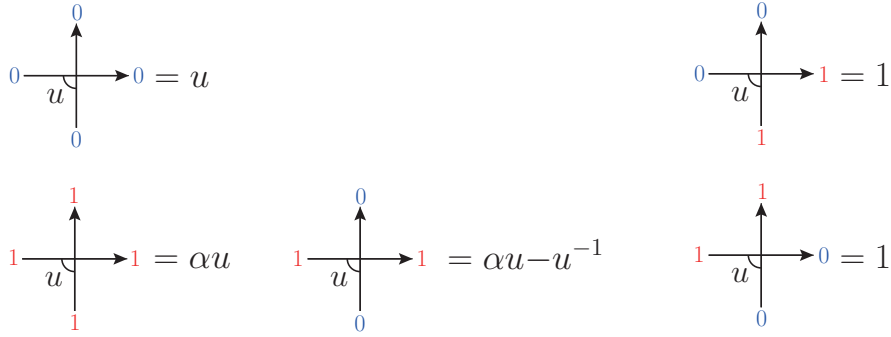}
\end{center}
\caption{The non-zero elements of the $L$-operator of the 
one-parameter family of five vertex models \eqref{loperator}.
The $L$-operator is pictorially represented as two crossing arrows.
The left (resp. up) arrow represents an auxiliary space 
(resp. a quantum space).
The indices 0 or 1 on the left (resp. right) of the vertices
denote the input (resp. output) states in the auxiliary space,
while those on the bottom (resp.  top) denote the 
input (resp. output) states in the quantum space.
}
\label{weight}
\end{figure}

\section{Scalar Products of state vectors}\label{scalar}
%
Here we construct a state vector of the integrable 
models defined in the preceding section by using
the quantum inverse scattering method (i.e the algebraic 
Bethe ansatz). The resultant
$N$-particle state $|\psi(\{u\}_N)\ket$ is 
characterized by $N$ unknown numbers $u_j\in \mathbb{C}$ 
($1\le j\le N$), which becomes an eigenstate of \eqref{Baxter} 
(or \eqref{TM}) if we choose the parameters $\{u\}_N$ as 
an arbitrary set of solutions of certain algebraic equation (i.e. 
the Bethe ansatz equation, see \eqref{BAE}). Hereafter we call the 
eigenstates the on-shell states, while we call the states with
arbitrary complex values of $\{u \}_N$ the off-shell states. In this
section, we construct the arbitrary off-shell states and show that 
their scalar products can be expressed as a determinant form.

First let us consider the monodromy matrix:
\begin{align}
T_{\mu}(u,\{w\})=\prod_{j=1}^M L_{\mu j}(u/w_j)
&=
\begin{pmatrix}
A(u,\{w\}) & B(u,\{w\})  \\
C(u,\{w\}) & D(u,\{w\})
\end{pmatrix}_{\mu}. 
\label{monodromy}
\end{align}
Here, for later convenience, we introduced the inhomogeneous parameters 
$w_1,\dots,w_M \in \mathbb{C}$. Taking the homogeneous limit 
$w_j\to 1$ ($1\le j\le M$), \eqref{monodromy1} is recovered:
\begin{align}
T(u,\{w\})|_{w_1=1,\dots,w_M=1}=T(u).
\label{homogeneous}
\end{align}
As in the above equation, hereafter we will omit $\{w\}$ for 
the quantities in the homogeneous limit (e.g. 
$A(u):=A(u,\{w\})|_{w_1=1,\dots,w_M=1}$).
The four elements of the monodromy matrix $A(u,\{w\})$, etc. are
the operators acting on the quantum space $V_1\otimes \dots \otimes V_M$.
The diagrammatic representations of these four elements are 
given by Figure.~\ref{abcdpic}.
\begin{figure}[ttt]
\begin{center}
\includegraphics[width=0.9\textwidth]{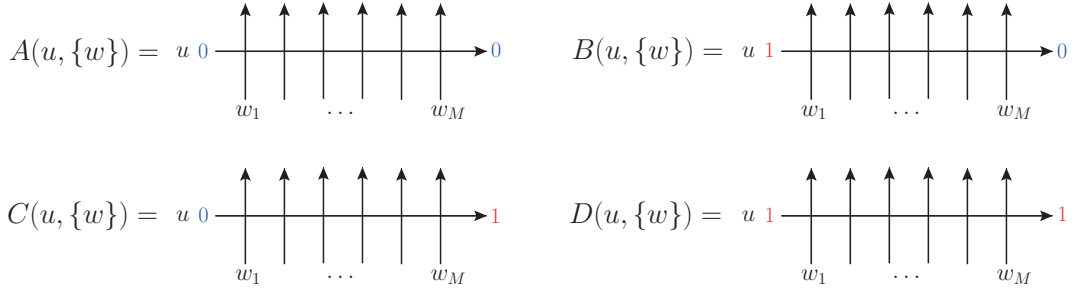}
\end{center}
\caption{The diagrammatic representation of 
the elements of the monodromy matrix \eqref{monodromy}
with the inhomogeneous parameters $w_1,\dots,w_M$.}
\label{abcdpic}
\end{figure}

Applying the $RLL$-relation \eqref{RLL} repeatedly,
the following intertwining relation
\begin{align}
R_{\mu\nu}(u,v)T_{\mu}(u,\{w\})T_{\nu}(v,\{w\})=
T_{\nu}(v,\{w\})T_{\mu}(u,\{w\})R_{\mu \nu}(u,v) \label{RTT}
\end{align}
follows. The relations listed below are obtained by the above equation,
which play a key role in the following calculations:
\begin{align}
&C(u,\{w\})B(v,\{w\})=g(u,v) 
\left[ A(u,\{w\})D(v,\{w\})-A(v,\{w\})D(u,\{w\}) \right], \nn  \\
&A(u,\{w\})B(v,\{w\})=f(u,v) B(v,\{w\}) A(u,\{w\})+
g(v,u)B(u,\{w\})A(v,\{w\}), \nn \\
&D(u,\{w\})B(v,\{w\})=f(v,u) B(v,\{w\}) D(u,\{w\})+
g(u,v)B(u,\{w\})D(v,\{w\}), \nn \\
&\left[B(u,\{w\}),B(v,\{w\})\right]=\left[C(u,\{w\}),C(v,\{w\})\right]=0.
\label{commutation}
\end{align}
The transfer matrix $\tau(u,\{w\})$ is then expressed as elements
of the monodromy matrix:
\begin{align}
\tau(u,\{w\})=\Tr_{W_\mu} T_{\mu}(u,\{w\})=A(u,\{w\})+D(u,\{w\}).
\label{transfer}
\end{align}

The arbitrary $N$-particle state $|\psi(\{u \}_N,\{w\}) \ket$
(resp. its dual $\langle \psi(\{u \}_N,\{w\})|$) 
(not normalized) with $N$ spectral parameters
$\{ u \}_N=\{ u_1,u_2,\dots,u_N \}$
is constructed by a multiple action
of $B$ (resp. $C$) operator on the vacuum state 
$|\Omega \rangle:=| 0^{M} \rangle:=|0\ket_1\
\otimes \dots \otimes |0\ket_M$
(resp. $\langle \Omega|:=\langle 0^{M}|:=
{}_1\bra 0|\otimes\dots \otimes{}_M\bra 0|$):
\begin{align}
|\psi(\{u \}_N,\{w\}) \rangle=\prod_{j=1}^N B(u_j,\{w\})| \Omega \rangle,
\quad
\langle \psi(\{u \}_N,\{w\})|=\langle \Omega| \prod_{j=1}^N
C(u_j,\{w\}).
\label{statevector}
\end{align}
Due to the commutativity of the operators $B$ or $C$ \eqref{commutation},
the states defined above (and also their scalar products) do not depend on the 
order of the product of $B$ or $C$.

By the standard procedure of the algebraic Bethe ansatz,
we have the followings.
\begin{proposition}\label{eigenstate}
The $N$-particle state $|\psi(\{u \}_N,\{w\}) \ket$
and its dual $\langle \psi(\{u \}_N,\{w\})$ become
an eigenstate (on-shell states) of the transfer matrix 
\eqref{transfer} when the set of parameters $\{u\}_N$
satisfies the Bethe ansatz equation:
\begin{align}
\frac{a(u_j,\{w\})}{d(u_j,\{w\})}=-\prod_{k=1}^N\frac{f(u_k,u_j)}{f(u_j,u_k)}
\label{BAE},
\end{align}
where
\begin{align}
a(u,\{w\})=\prod_{j=1}^M \frac{u}{w_j}, \quad
d(u,\{w\})=\prod_{j=1}^M \left( \frac{\alpha u}{w_j}-\frac{w_j}{u} \right).
\end{align}
Then the eigenvalue of the transfer matrix is given by
\begin{align}
\tau(u,\{w\})=a(u,\{w\})\prod_{j=1}^N f(u,u_j)
             +d(u,\{w\})\prod_{j=1}^N f(u_j,u).
\label{EV}
\end{align}
\end{proposition}

The scalar product between the arbitrary off-shell state vectors, 
which is mainly considered in this section, is defined as
\begin{align}
\langle \psi(\{ u \}_N,\{w\})| \psi(\{ v \}_N,\{w\}) \rangle
=\langle \Omega| \prod_{j=1}^N C(u_j,\{w\})
\prod_{k=1}^N B(v_k,\{w\})| \Omega \rangle
\label{SP}
\end{align}
with $u_j, v_k\in \mathbb{C}$.
In the homogeneous limit $w_j\to 1$ ($1\le j \le N$), the following theorem is 
known \cite{Bo,KBI,Ta}.
\begin{theorem}{\label{scalarthm}}
The scalar product \eqref{SP} in the homogeneous limit 
$w_j\to1$ ($1\le j \le N$) is given by a determinant form:
\begin{align}
\bra \psi(\{ u \}_N)| \psi(\{ v \}_N) \ket
=
\prod_{1 \le j <k \le N} \frac{1}{(u_j^2-u_k^2)(v_k^2-v_j^2)}
\mathrm{det}_N Q(\{ u \}_N|\{ v \}_N),
\label{generalscalar}
\end{align}
where $\{u\}_N$ and $\{v\}_N$ are arbitrary sets
of complex values (i.e. off-shell conditions), and
$Q$ 
is an $N \times N$ matrix with matrix elements
\begin{align}
Q(\{ u \}_N|\{ v \}_N)_{jk}=
\frac{ a(u_j) d(v_k)v_k^{2(N-1)}
           -a(v_k) d(u_j)u_j^{2(N-1)}
         }
        {v_k/u_j-u_j/v_k}.
\label{element}
\end{align}
\end{theorem}

Here we will show the above determinant formula by
utilizing a method recently developed  by Wheeler
in the calculation of the scalar product of the
spin-1/2 XXZ chain \cite{Wh}. This technique is based on the Izergin-Korepin 
procedure \cite{Ko,Iz}, which is originally a method to calculate the domain 
wall boundary partition function of the six vertex model \cite{Ko,Iz}.  
In contrast to the spin-1/2 XXZ chain, 
in our case there is no need to impose the Bethe ansatz equation
(i.e. on-shell condition) to show the determinant formula. 
In other words, the determinant formula \eqref{generalscalar} is valid for arbitrary 
off-shell states.

What plays a fundamental role in this method is the
following intermediate scalar products (see also Figure \ref{intermediatepic}
for a diagrammatic representation)
\begin{align}
S(\{ u \}_n | \{ v \}_N| \{ w \})=
\langle 
 0^{M-N+n} 1^{N-n} | \prod_{j=1}^n C(u_j,\{w\})
                                  \prod_{k=1}^N B(v_k,\{w\}) 
 |\Omega 
\rangle. 
\label{intermediatedef}
\end{align}
The term ``intermediate" stems from the fact that \eqref{intermediatedef}
interpolates the scalar product ($n=N$) and the domain wall boundary 
partition function ($n=0$).
\begin{figure}[tt]
\begin{center}
\includegraphics[width=0.6\textwidth]{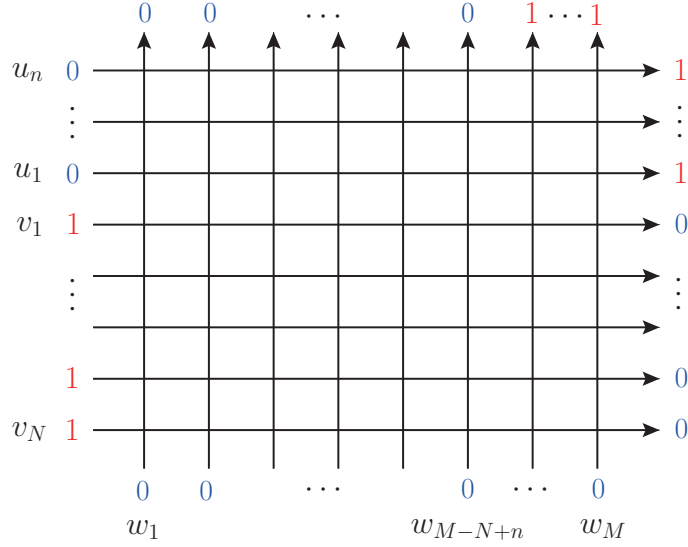}
\end{center}
\caption{The graphical representation of the intermediate scalar products 
\eqref{intermediatedef}
with inhomogeneous parameters $\{w\}$.
The case $n=N$ corresponds to the usual scalar product \eqref{SP}, while the case 
$n=0$ corresponds to the domain wall boundary partition function.}
\label{intermediatepic}
\end{figure}
We have the following lemma regarding the properties of the intermediate 
scalar product.
\begin{lemma}{\label{property}}
The intermediate scalar product \eqref{intermediatedef} 
$S(\{ u \}_n | \{ v \}_N| \{ w \})$ satisfies the
following properties. 
\begin{enumerate}
\item $S(\{ u \}_n | \{ v \}_N| \{ w \})$ is symmetric with respect to
the variables $\{ w_1,\dots,w_{M-N+n} \}$. 
\item $\prod_{j=1}^n u_j^{M+2n-2N-1} S(\{ u \}_n | \{ v \}_N| \{ w \})$ 
is a polynomial of degree $M-N+n-1$ in $u_n^2$. 
\item
The following recursive relations between the intermediate scalar products
hold
\begin{align}
&S(\{ u \}_n|\{ v \}_N|\{ w \})|_{u_n=\pm \alpha^{-1/2} w_{M-N+n}} \nn \\
& \qquad \qquad 
=\alpha^{N-n-(M-1)/2} (\pm 1)^{M-1} \frac{w_{M+n-N}^M}{\prod_{j=1}^M w_j}
S(\{ u \}_{n-1}|\{ v \}_N|\{ w \}). \label{recursiverelation}
\end{align}
\item  The case $n=0$ of the intermediate scalar products has the 
following form:
\begin{align}
S(\{ u \}_0|\{ v \}_N|\{ w \})
=\alpha^{N(N-1)/2} \prod_{j=1}^N \prod_{k=1}^{M-N}
\left( \frac{\alpha v_j}{w_k}-\frac{w_k}{v_j} \right)
\frac{\prod_{j=1}^N v_j^{N-1}}{\prod_{j=M-N+1}^M w_j^{N-1}}.
\label{initial}
\end{align}
\end{enumerate}
\end{lemma}
\begin{proof}
Property 1 follows from the $RLL$-relation
\begin{align}
\widetilde{R}_{jk}(w_j/w_k)L_{\mu k}(u/w_k)L_{\mu j}(u/w_j)
=L_{\mu j}(u/w_j)L_{\mu k}(u/w_k)\widetilde{R}_{jk}(w_j/w_k)
\label{RLL2}
\end{align}
holding in 
$\End(W_{\mu} \otimes V_j \otimes V_k$). Here $\widetilde{R}$
is given by
\begin{align}
\widetilde{R}(u)
=
\begin{pmatrix}
u & 0 & 0 & 0 \\
0 & 0 & 1 & 0 \\
0 & 1 & \alpha(u-u^{-1}) & 0 \\
0 & 0 & 0 & u
\end{pmatrix},
 \label{Rtildematrix}
\end{align}
which intertwines the $L$-operators acting on a common auxiliary space
(but acting on  different quantum spaces).
Note the usual $RLL$-relation \eqref{RLL2} intertwines
the $L$-operators acting on a same quantum space but acting on different auxiliary spaces.
The above $RLL$-relation \eqref{Rtildematrix} allows one to construct
the monodromy matrix as a product of the $L$-operators
acting on the same quantum space (see also the next section), and rewriting
the intermediate scalar products in terms of the
resultant monodromy matrices makes one see Property 1 holds. 

Property 2 can be shown by inserting the completeness relation
into the intermediate scalar products (see Figure \ref{genericrecursivepic} for
a graphical interpretation)
\begin{align}
S(\{ u \}_n | \{ v \}_N| \{ w \})&=\langle 
0^{M-N+n} 1^{N-n} |
\prod_{j=1}^n C(u_j,\{w\}) \prod_{k=1}^N B(v_k,\{w\})|\Omega \rangle
\nonumber \\
&=\sum_{k=1}^{M-N+n}
\langle 
0^{M-N+n} 1^{N-n} |C(u_n,\{w\})
| 0^{k-1} 1 0^{M-N+n-k} 1^{N-n} 
\rangle
\nonumber \\
&\qquad\times
\langle 0^{k-1} 1 0^{M-N+n-k} 1^{N-n}|
\prod_{j=1}^{n-1}C(u_{j},\{w\}) \prod_{k=1}^N B(v_k,\{w\})|\Omega \rangle,
\label{genericrecursiverelation}
\end{align}
and noting the factor containing $u_n$ is calculated as
\begin{align}
&\langle 
0^{M-N+n} 1^{N-n} |
C(u_n,\{w\})
| 0^{k-1} 1 0^{M-N+n-k} 1^{N-n} \rangle \nn \\
& \qquad \qquad \qquad \qquad=\frac{\alpha^{N-n} u_n^{N-n+k-1}}
{\prod_{j=1}^{k-1} w_j \prod_{j=M-N+n+1}^M w_j}
\prod_{j=k+1}^{M-N+n} \left( \frac{\alpha u_n}{w_j}-\frac{w_j}{u_n} \right).
\end{align}
\begin{figure}[tt]
\begin{center}
\includegraphics[width=0.9\textwidth]{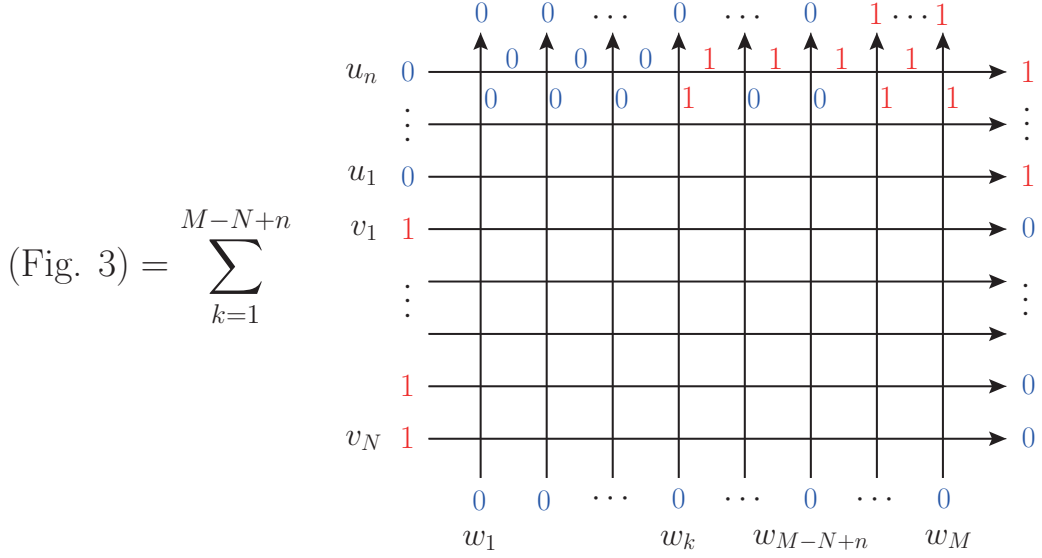}
\end{center}
\caption{The intermediate scalar products
where the completeness relation is inserted \eqref{genericrecursiverelation}. 
Note the parameter $u_n$ comes only from the top row.}
\label{genericrecursivepic}
\end{figure}

Property 3 can be obtained by setting $u_n=\pm \alpha^{-1/2} w_{M-N+n}$
in \eqref{genericrecursiverelation},
or can be directly observed by its graphical representation
(Figure \ref{recursivepic}) that the top row is completely frozen. 
\begin{figure}[tt]
\begin{center}
\includegraphics[width=0.7\textwidth]{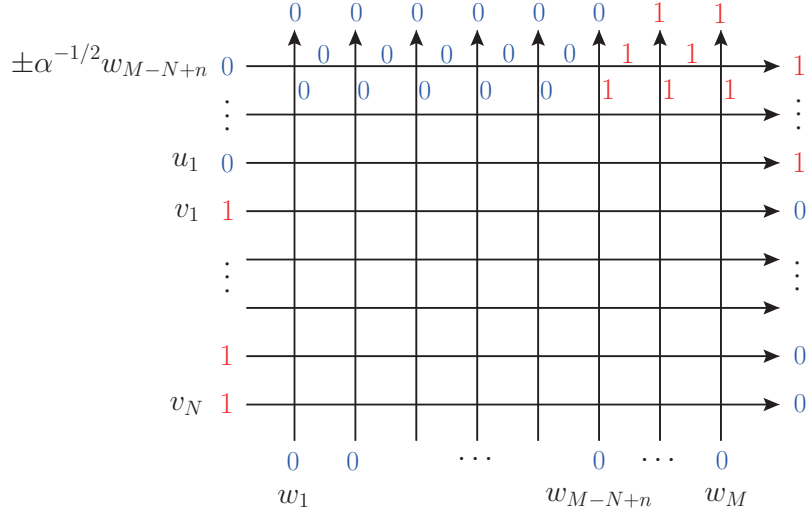}
\end{center}
\caption{The graphical representation of the 
recursive relation \eqref{recursiverelation}. We can see that the top row is
frozen by setting the spectral parameter $u_n$ to
$u_n=\pm \alpha^{-1/2} w_{M-N+n}$.}
\label{recursivepic}
\end{figure}

Property 4 can be shown by noting that all the internal states
are frozen (Figure \ref{initialpic}),
and reading out and multiplying all the weights of the $L$-operators
to find \eqref{initial}.
\end{proof}

\begin{figure}[tt]
\begin{center}
\includegraphics[width=0.7\textwidth]{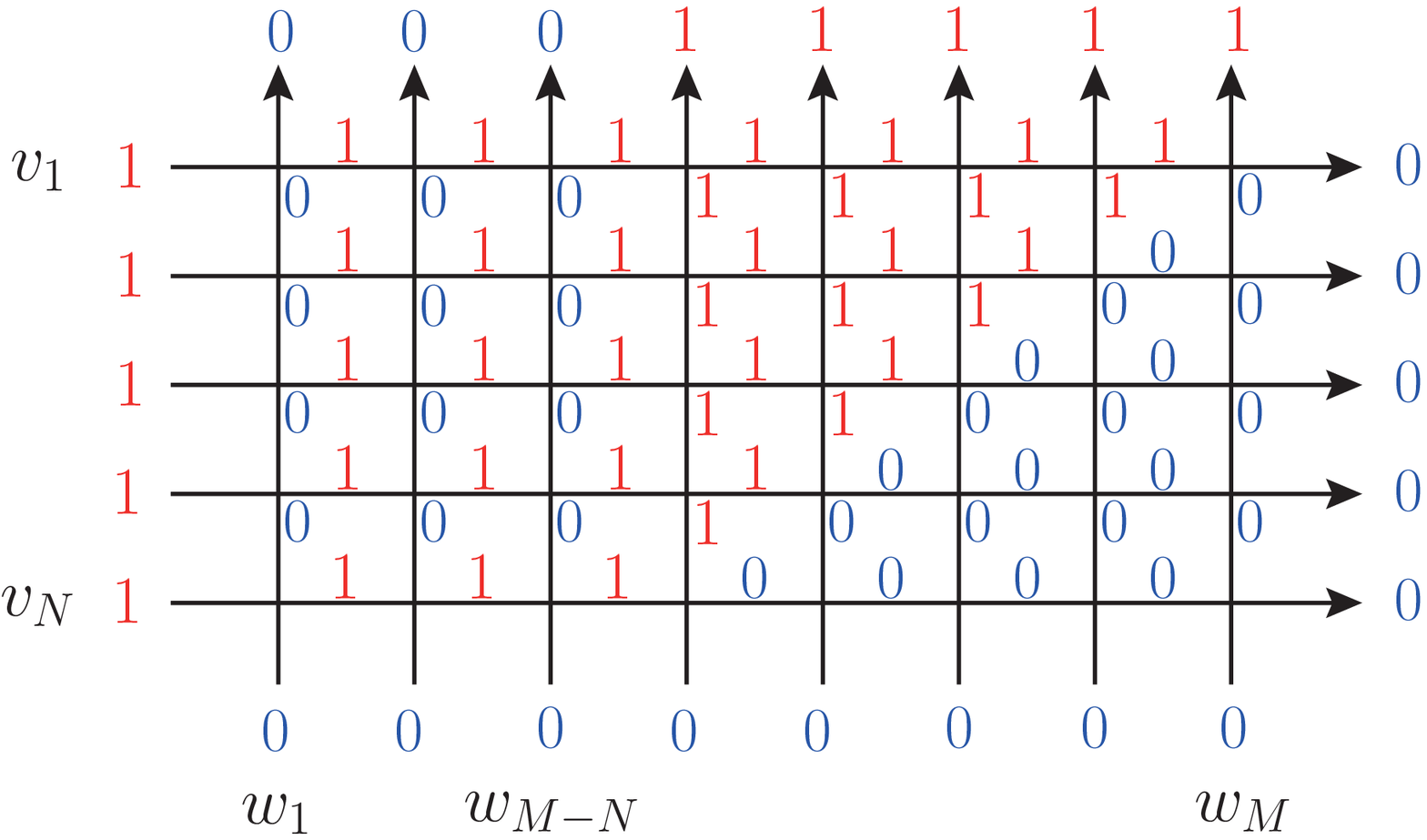}
\end{center}
\caption{The intermediate scalar products \eqref{initial} for $n=0$,
which corresponds to the domain wall boundary partition function.
One sees all the internal states are frozen when the boundary states
are fixed to the configuration in the figure.}
\label{initialpic}
\end{figure}

\begin{lemma}\label{uniqueness}
The properties in Lemma~\ref{property} uniquely determine
the intermediate scalar product \eqref{intermediatedef}.
\end{lemma}
\begin{proof}
The proof is by induction on $n$. For $n=0$, by Property 4 
the assertion is trivial. Assume by induction that the assertion 
holds for $n-1$. Taking into account Property 1, one
finds that Property 3 gives values of $\prod_{j=1}^n u_j^{M+2n-2N-1}
S(\{u\}_{n}|\{v\}_N|\{w\})$
at $M-N+n$ distinct points of $u_n^2$. By this  together with Property 2,
$S(\{u\}_{n}|\{v\}_N|\{w\})$ is uniquely determined. Thus 
the assertion holds for $n$.
\end{proof}

Due to Lemma~\ref{uniqueness}, the following determinant
representation for the intermediate scalar product is valid.
\begin{theorem}{\label{determinantthm}}
The intermediate scalar product $S(\{ u \}_n|\{ v \}_N|\{ w \})$ \eqref{intermediatedef} 
has the following determinant form:
\begin{align}
S(\{ u \}_n|\{ v \}_N|\{ w \}) 
=&
\prod_{M-N+n+1 \le j<k \le M} \frac{1}{w_j^2-w_k^2}
\prod_{1 \le j<k \le n} \frac{1}{u_j^2-u_k^2}
\prod_{1 \le j<k \le N} \frac{1}{v_k^2-v_j^2} \nn \\
&\times
\mathrm{det}_N Q(\{ u \}_n|\{ v \}_N|\{ w \})
\label{intermediatedeterminant}
\end{align}
with an $N \times N$ matrix
$Q(\{ u \}_n|\{ v \}_N|\{ w \})$ whose matrix elements
are given by
\begin{align}
&Q(\{ u \}_n|\{ v \}_N|\{ w \})_{jk} \nn \\
&\quad =
\begin{cases}
\frac{\displaystyle
a(u_j,\{w\})d(v_k,\{w\})v_k^{2(N-1)}
-a(v_k,\{w\})d(u_j,\{w\}) u_j^{2(N-1)}
}
{\displaystyle
(v_k/u_j-u_j/v_k)
\prod_{l=M-N+n+1}^M (u_j^2-\alpha^{-1} w_l^2)
},  & \text{ ($1\le j \le n$)}  \\
\displaystyle
v_k^{2(N-1)} \prod_{\substack{l=1 \\ l \neq M-N+j}}^M
\left( \frac{\alpha v_k}{w_l}-\frac{w_l}{v_k} \right),
& \text{ ($n+1\le j \le N$)}
\end{cases}.
\end{align}
\end{theorem}
\begin{proof}
We can directly see that the determinant formula \eqref{intermediatedeterminant}
satisfies all the properties in Lemma~\ref{property}. 
To show Property 2, we just use the fact that the singularities $u_n^2=u_j^2$ ($1\le j\le n-1$)
in the prefactor, and $u_n^2=\alpha^{-1}w_j^2$ ($M-N+n+1\le j \le M$) 
and $u_n^2=v_j^2$ ($1\le j\le n$) in elements of the determinant are removal.  
For Property 4, we utilize  the Cauchy determinant formula to obtain
\begin{align}
\mathrm{det}_N \left\{ \left( \frac{\alpha^{1/2} v_k}{w_{M-N+j}}
-\frac{w_{M-N+j}}{\alpha^{1/2} v_k} \right)^{-1} \right\}
=\frac{
\displaystyle
\prod_{M-N+1 \le j<k \le M} 
\left( \frac{w_j}{w_k}-\frac{w_k}{w_j}
\right)
\prod_{1 \le j<k \le N}
\left(
\frac{v_k}{v_j}-\frac{v_j}{v_k}
\right)
}
{
\displaystyle
\prod_{j=1}^N \prod_{k=M-N+1}^M
\left(
\frac{\alpha^{1/2} v_j}{w_k}-\frac{w_k}{\alpha^{1/2} v_j}
\right)
}.
\end{align}
Finally due to Lemma~\ref{uniqueness},   the determinant formula 
\eqref{intermediatedeterminant} holds.
\end{proof}
\begin{corollary}
Taking $n=N$ in \eqref{intermediatedeterminant} yields the determinant
representation of the scalar product for the five vertex model with inhomogeneous
parameters \eqref{SP}:
\begin{align}
\bra \psi(\{ u \}_N,\{w\})| \psi(\{ v \}_N,\{w\}) \ket
=&
\prod_{1 \le j<k \le n} \frac{1}{(u_j^2-u_k^2)(v_k^2-v_j^2)} 
\mathrm{det}_N Q(\{ u \}_N|\{ v \}_N|\{ w \})
\end{align}
with
\begin{align}
Q(\{ u \}_N|\{ v \}_N|\{ w \})_{jk} 
=
\frac{
a(u_j,\{w\})d(v_k,\{w\})v_k^{2(N-1)}
-a(v_k,\{w\})d(u_j,\{w\}) u_j^{2(N-1)}
}
{
v_k/u_j-u_j/v_k
}.
\end{align}
Further taking the homogeneous limit $w_j\to 1$ ($1\le j \le n$) 
yields \eqref{generalscalar} in 
Theorem~\ref{determinantthm}. 
\end{corollary}
The state vectors $| \psi(\{u\}_N)\ket$
and $\bra  \psi(\{u\}_N)|$
become the energy eigenstates of \eqref{markov}, when
an arbitrary set of solutions $\{u\}_N$ to the Bethe ansatz 
equation \eqref{BAE} in the homogeneous limit $w_j\to1$ ($1\le j \le N$)
is substituted into the state vectors. 
Then we have the following corollary regarding the 
norm of the eigenstates. 
\begin{corollary}
The norm of the eigenstates in the homogeneous limit $w_j\to 1$ ($1\le j \le n$) 
is given by
\begin{align}
\bra \psi(\{u\}_N) | \psi(\{u\}_N) \ket=
\prod_{j=1}^N u_j^{2(M+N-1)}
\prod_{\substack{ j,k=1 \atop j \neq k}}^N \ \frac{1}{u_j^2-u_k^2}
\mathrm{det}_N \widetilde{Q}(\{u\}_N)
\label{norm}
\end{align}
with 
\begin{align}
\widetilde{Q}_{jk}(\{u\}_N)= -1+\frac{\alpha N+(M-N) u_j^{-2}}{\alpha-u_j^{-2}}\delta_{jk}.
\end{align}
By use of Sylvester's determinant theorem,  the 
determinant in the above  further reduces to
\begin{align}
\mathrm{det}_N \widetilde{Q}=
\prod_{j=1}^N \frac{\alpha N+(M-N)u_j^{-2}}{\alpha-u_j^{-2}}
\left(1-\sum_{j=1}^N \frac{\alpha-u_j^{-2}}{\alpha N +(M-N)u_j^{-2}} \right).
\end{align}
\end{corollary}
%
\section{Wavefunctions}
%
In this section, we compute the overlap between an arbitrary off-shell 
$N$-particle state $|\psi(\{u\}_N)\ket$ and 
the (normalized) state with an arbitrary particle configuration 
$|x_1 \cdots x_N\ket$ $(x_1<\dots<x_N$), 
where $x_j$ denotes the positions of the particles. Namely here we evaluate the 
wavefunction
$\bra x_1 \cdots x_N | \psi(\{u\}_N) \ket$ and its dual 
$\bra \psi(\{u\}_N)|x_1\dots x_N \ket$.
One finds these quantities are crucial to describe physically interesting phenomena 
such as the relaxation dynamics as in Section~\ref{TASEP},  
because the state $|\psi(\{u\})_N\ket$ becomes an
eigenstate of the Hamiltonian \eqref{Baxter} (correspondingly
$\bra x_1\dots x_N | \psi(\{u\}_N)\ket$ becomes an energy eigenfunction), 
if we choose $\{u\}_N$ as 
an arbitrary set of solutions of the 
Bethe ansatz equation (see Proposition~\ref{eigenstate} and Section~\ref{TASEP} 
for details). Here and in what follows, we
consider the homogeneous case $w_1=1,\dots,w_M=1$, and as noted in the previous
section we omit $\{w\}$ as in \eqref{homogeneous}.

The main results in this section are summarized in the following theorem.
\begin{theorem}\label{overlapthm}
The wavefunctions can be written as the following determinant formulae:
\begin{align}
&\bra x_1 \cdots x_N |  \psi(\{v\}_N) \ket=
\frac{\prod_{j=1}^N v_j^{M-1} (\alpha v_j^2-1)^{-1}}
{\prod_{1 \le j < k \le N}(v_k^2-v_j^2)}
\mathrm{det}_N(v_j^{2k}(\alpha-v_j^{-2})^{x_k}), \label{generaloverlaptwo} \\
&\bra \psi(\{u\}_N)|x_1 \cdots x_N \rangle
=\frac{\prod_{j=1}^N (\alpha u_j-u_j^{-1})^M u_j^{2N-1}}
{\prod_{1 \le j < k \le N}
(u_j^2-u_k^2)} \mathrm{det}_N(u_j^{-2k} (\alpha-u_j^{-2})^{-x_k}), \label{generaloverlap} 
\end{align}
where 
\begin{align}
   \bra x_1\dots x_N|=\bra \Omega|\prod_{j=1}^N \sigma^+_{x_j}, \qquad
|x_1\dots x_N\ket=\prod_{j=1}^N \sigma^-_{x_j} |\Omega\ket,
\end{align}
and $\{v\}_N$ and $\{u\}_N$  are  sets of arbitrary complex parameters.
\end{theorem}
The strategy to  show Theorem~\ref{overlapthm} is as follows.
We first rewrite the wavefunctions
into a matrix product form, following \cite{GMmat}.
The matrix product form can be expressed as a determinant with some overall
factor which remains to be calculated. The information of the particle configuration 
$\{x_1,x_2,\dots,x_N \}$ is encoded in the determinant.
On the other hand, the overall factor is independent of the
particle positions, and therefore we can  determine this factor by
considering the specific configuration: we 
explicitly calculate it with the help of the result 
for the overlap of the consecutive configuration (i.e. $x_j=j$) obtained
in \cite{MSS,MSS2}.

Let us begin to compute the wavefunctions. We consider 
\eqref{generaloverlap} first. The proof of \eqref{generaloverlaptwo} can 
be done in a similar way. First we shall rewrite the wavefunction 
$\bra \psi(\{u\}_N| x_1\dots x_N \ket$ into the matrix product 
representation.
With the help of graphical description,
one finds that the wavefunction can be written as
\begin{align}
\bra \Omega|\prod_{j=1}^N C(u_j)|x_1\dots x_N \ket
=\Tr_{W^{\otimes N}}
\left[
\bra \Omega|\prod_{\mu=1}^N T_{\mu}(u_\mu) |x_1 \cdots x_N \ket P
\right],
\label{overlap}
\end{align}
where $P=| 0^N \rangle \langle 1^N |$
is an operator acting on the tensor product of auxiliary spaces
$W_1\otimes  \dots \otimes W_N$.
The trace here is also over the auxiliary spaces.
Due to the commutativity of the operators $B$ or $C$ \eqref{commutation},
the wavefunctions do not depend on the order of the product of $B$ or $C$.
In other words, the wavefunctions are symmetric with respect to
the parameters $\{u\}_N$ or $\{v\}_N$.
Changing the viewpoint of the products of the monodromy matrices, we have
\begin{align}
\prod_{\mu=1}^N T_{\mu}(u_\mu)
=\prod_{j=1}^M \mathcal{T}_j(\{u\}_N),
\end{align}
where $\mathcal{T}_j(\{u\}_N):= 
\prod_{\mu=1}^N L_{\mu j}(u_{\mu}) \in \End( W^{\otimes N} \otimes V_j)$
can be regarded as a monodromy matrix consisting of
$L$-operators acting on the same quantum space $V_j$
(but acting on different auxiliary spaces).  The monodromy matrix
is decomposed as
\begin{align}
\mathcal{T}_j(\{u\}_N)
&:=\begin{pmatrix}
\mathcal{A}_N (\{u\}_N) & \mathcal{B}_N(\{u\}_N) \\
\mathcal{C}_N(\{u\}_N) &  \mathcal{D}_N(\{u\}_N)
\end{pmatrix}_j ,
\label{decomp}
\end{align}
where the elements
($\mathcal{A}_N$, etc.) act on 
$W_1\otimes \dots \otimes W_N$.
The wavefunction \eqref{overlap} can then be rewritten by 
$\mathcal{T}_j(\{u\}_N)$ as
\begin{align}
\bra \psi(\{u\}_N)|x_1\dots x_N \ket
&=\Tr_{W^{\otimes N}}\left[
\langle \Omega| \prod_{j=1}^M \mathcal{T}_j(\{u\}_N)|x_1 \cdots x_N \rangle  P
\right] \nn \\
&=\Tr_{W^{\otimes N}}\left[
\mathcal{A}_N^{M-x_N}
\mathcal{B}_N
\mathcal{A}_N^{x_N-x_{N-1}-1}
\dots\mathcal{B}_N\mathcal{A}_N^{x_2-x_1-1}\mathcal{B}_N\mathcal{A}_N^{x_1-1} P\right].
\label{reov}
\end{align}
In Figure~\ref{overlapelementspic}, we depict the elements  $\mathcal{A}_n(\{u\}_n)$
and $\mathcal{B}_n(\{u\}_n)$ of the monodromy matrix $\mathcal{T}_j(\{u\}_n)$,
which explicitly appear in \eqref{reov}. 
\begin{figure}[tt]
\begin{center}
\includegraphics[width=0.6\textwidth]{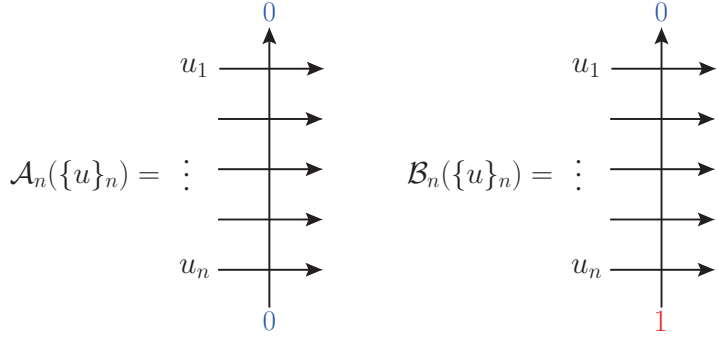}
\end{center}
\caption{The elements $\mathcal{A}(\{u\}_n)$ and 
$\mathcal{B}(\{u\}_n)$ of the monodromy matrix $\mathcal{T}_j(\{u \}_n)$ \eqref{decomp}.}
\label{overlapelementspic}
\end{figure}

For these operators, one finds the following recursive 
relations: 
\begin{align}
&\mathcal{A}_{n+1}(\{u\}_{n+1})
=\mathcal{A}_n(\{u\}_n) \otimes
\begin{pmatrix}
u_{n+1} & 0  \\
0 & \alpha u_{n+1}-u_{n+1}^{-1}
\end{pmatrix}
+\mathcal{B}_n(\{u\}_n) \otimes
\begin{pmatrix} 
0 & 1  \\
0 & 0
\end{pmatrix},   \label{reop1} \\
&\mathcal{B}_{n+1}(\{u\}_{n+1})
=\mathcal{A}_n(\{u\}_n) \otimes
\begin{pmatrix}
0 & 0  \\
1 & 0
\end{pmatrix}
+\mathcal{B}_n(\{u\}_n) \otimes
\begin{pmatrix}
0 & 0  \\
0 & \alpha u_{n+1}
\end{pmatrix} \label{reop2}  
\end{align}
with the initial condition
\begin{align}
\mathcal{A}_1=
\begin{pmatrix}
u_1 & 0 \\
0    &  \alpha u_1-u_1^{-1}
\end{pmatrix}, \quad
\mathcal{B}_1=
\begin{pmatrix}
0 & 0 \\
1 & 0 
\end{pmatrix}.
\label{initialAB}
\end{align}
See Figure~\ref{overlaprecpic} for a graphical description
of the recursion relation for the operator $\mathcal{A}_{n}(\{u\}_n)$.
\begin{figure}[ttt]
\begin{center}
\includegraphics[width=0.6\textwidth]{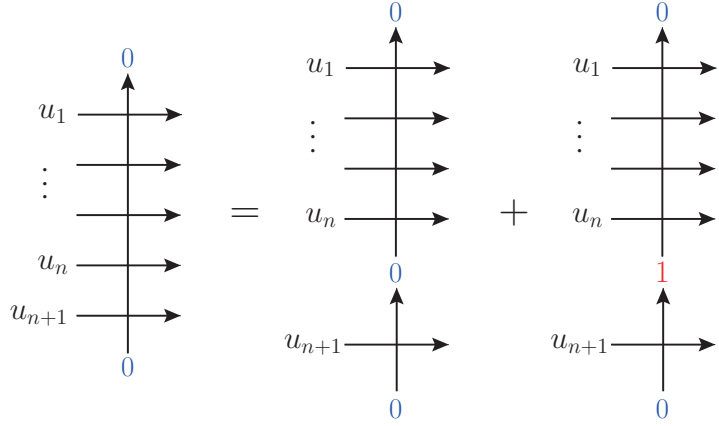}
\end{center}
\caption{The graphical description of the recursive relation 
for the element $\mathcal{A}_n(\{u\}_n)$ (see \eqref{reop1}).}
\label{overlaprecpic}
\end{figure}
By using the recursive relations \eqref{reop1} and \eqref{reop2},
one sees that these operators satisfy the following simple algebra.
\begin{lemma}\label{algebra}
There exists a decomposition of $\mathcal{B}_n$ :
$\mathcal{B}_n=\sum_{j=1}^n \mathcal{B}_n^{(j)}$ such that
the following algebraic relations hold for $\mathcal{A}_n$ and $\mathcal{B}_n^{(j)}$:
\begin{align}
&\mathcal{B}_n^{(j)}\mathcal{A}_n
=\frac{u_j}{\alpha u_j-u_j^{-1}}\mathcal{A}_n \mathcal{B}_n^{(j)}, \label{rel2} \\
&(\mathcal{B}_n^{(j)})^2=0, \label{rel3} \\
&(\alpha u_j^2-1)\mathcal{B}_n^{(j)}\mathcal{B}_n^{(k)}
=-(\alpha u_k^2-1)\mathcal{B}_n^{(k)}\mathcal{B}_n^{(j)}, \ \ \ (j \neq k)
\label{rel4}.
\end{align}
\end{lemma}
\begin{proof}
This can be shown by induction on $n$.  For $n=1$,  from \eqref{initialAB}
$\mathcal{A}_1$ is diagonal and one directly sees that the relations are valid.
For $n$, we assume that $\mathcal{A}_n$ is diagonalizable and write the 
corresponding diagonal matrix as $\mathscr{A}_n=G_n^{-1}\mathcal{A}_n G_n$. 
Also writing  $\mathscr{B}_n=G_n^{-1} \mathcal{B}_n G_n$ and 
$\mathscr{B}_n=\sum_{j=1}^{n} \mathscr{B}_n^{(j)}$, and  noting
the algebraic relations above do not depend on the choice of basis, we suppose by the
induction hypothesis that the same relations are satisfied by $\mathscr{A}_n$
and $\mathscr{B}_n^{(j)}$. 

Now we shall show that they also hold for $n+1$. To this end, first we 
construct $G_{n+1}$. Noting from \eqref{reop1} that $\mathcal{A}_{n+1}$ is an 
upper triangular block matrix whose block diagonal elements are written in 
terms of $\mathcal{A}_n$, 
we assume that $G_{n+1}$ is written as
\begin{equation}
G_{n+1}=
\begin{pmatrix}
G_n &  G_n H_n \\
0  & G_n
\end{pmatrix},
\label{G-matrix}
\end{equation}
where $2n\times 2n$ matrix $H_n$ remains to be determined. 
Using the induction hypothesis for $n$, one obtains
\begin{equation}
G_{n+1}^{-1}\mathcal{A}_{n+1} G_{n+1}=
\begin{pmatrix}
u_{n+1} \mathscr{A}_n & u_{n+1}\mathscr{A}_n H_n+\mathscr{B}_n
                                    -(\alpha u_{n+1}-u_{n+1}^{-1})H_n \mathscr{A}_n \\
0 & (\alpha u_{n+1}-u_{n+1}^{-1})\mathscr{A}_n
\end{pmatrix}.
\end{equation}
The above matrix is guaranteed to be diagonal when 
\begin{equation}
 \mathscr{B}_n=(\alpha u_{n+1}-u_{n+1}^{-1})H_n \mathscr{A}_n-
                         u_{n+1}\mathscr{A}_n H_n.
\end{equation}
Utilizing the above relation and  recalling  $\mathscr{A}_n$
and $\mathscr{B}^{(j)}_n$ satisfy the relation same as that in \eqref{rel2}, 
one finds
\begin{align}
H_n=\mathscr{A}^{-1}_n\sum_{j=1}^n 
\frac{(\alpha u_j-u_j^{-1}) }
        {u_j^{-1}u_{n+1}-u_j u_{n+1}^{-1}} \mathscr{B}_n^{(j)}.
\label{H-matrix}
\end{align}
One thus obtains the diagonal matrix $\mathscr{A}_{n+1}$:
\begin{align}
\mathscr{A}_{n+1}=
\begin{pmatrix}
u_{n+1}\mathscr{A}_n & 0 \\
0 & (\alpha u_{n+1}-u_{n+1}^{-1})\mathscr{A}_n
\end{pmatrix}.
\label{A-matrix}
\end{align}
The remaining task is to derive  $\mathscr{B}_{n+1}^{(j)}$ and
to prove the relations \eqref{rel2}--\eqref{rel4} hold for $n+1$.
Combining  \eqref{reop2}, \eqref{G-matrix} and \eqref{H-matrix},
and also inserting the relations \eqref{rel3} and \eqref{rel4},
one arrives at $\mathscr{B}_{n+1}=\sum_{j=1}^{n+1}\mathscr{B}_{n+1}^{(j)}$
where
\begin{align}
\mathscr{B}_{n+1}^{(j)}=
\begin{cases} \displaystyle
\frac{1}{u_j u_{n+1}^{-1}-u_j^{-1}u_{n+1}}
\begin{pmatrix}
u_j \mathscr{B}_n^{(j)} & 0 \\
0 & u_j^{-1}(1-\alpha u_{n+1}^2) \mathscr{B}_n^{(j)}
\end{pmatrix}  & \text{ for $1\le j \le n$} \\[6mm]
\begin{pmatrix}
0  & 0 \\
\mathscr{A}_n & 0
\end{pmatrix}  & \text{ for $j=n+1$}
\end{cases}.
\label{B-matrix}
\end{align}
Finally recalling that $\mathscr{A}_n$ and $\mathscr{B}_n^{(j)}$ 
are supposed to
satisfy the relations \eqref{rel2}--\eqref{rel4} and using the explicit
form of $\mathscr{A}_{n+1}$ \eqref{A-matrix} and $\mathscr{B}_{n+1}^{(j)}$ 
\eqref{B-matrix}, one sees they satisfy the same algebraic relations as those 
in \eqref{rel2}--\eqref{rel4} for $n+1$.
\end{proof}

Due to the algebraic relations \eqref{rel2} and \eqref{rel3} in Lemma~\ref{algebra}, 
the matrix product form for the wavefunction \eqref{reov} can be rewritten as
\begin{align}
\bra \psi(\{u\}_N)|x_1\dots x_N \ket
 = \sum_{\sigma \in \mathfrak{S}_N} &
    \prod_{j=1}^N
  u_{\sigma(j)}^{-2(M-N)}\left(\alpha u^2_{\sigma(j)}-1\right)^{M-N+j}
  u_{\sigma(j)}^{-2 j} \left(\alpha-u_{\sigma(j)}^{-2}\right)^{-x_j} \nn \\
& \times
\Tr_{W^{\otimes N}}\left[
\mathcal{B}_N^{(\sigma(N))}
\dots\mathcal{B}_N^{(\sigma(1))}\mathcal{A}_N^{M-N} P \right],
\end{align}
where $\mathfrak{S}_n$ is the symmetric group of order $N$.
Using \eqref{rel4} to arrange the order of the matrix product
$\mathcal{B}_N^{(\sigma(N))}\dots\mathcal{B}_N^{(\sigma(1))}$
in the canonical order
$\mathcal{B}_N^{(N)}\dots\mathcal{B}_N^{(1)}$ yields the 
following
determinant form:
\begin{align}
\bra \psi(\{u\}_N)|x_1\dots x_N \ket
&= K \sum_{\sigma \in \mathfrak{S}_n} (-1)^{\sigma} 
\prod_{j=1}^N
u_{\sigma(j)}^{-2 j}\left(\alpha-u_{\sigma(j)}^{-2}\right)^{-x_j} \nn \\
&=K \mathrm{det}_N 
\left[  u_{j}^{-2 k}  \left(\alpha-u_{j}^{-2}\right)^{-x_k}
\right],
\label{predet}
\end{align}
where the prefactor $K$ given below remains to be determined:
\begin{align}
K=\prod_{j=1}^N\left(\alpha u^2_{j}-1\right)^{M-N+j}u_{j}^{-2(M-N)}
\Tr_{W^{\otimes N}}\left[
\mathcal{B}_N^{(N)}
\dots\mathcal{B}_N^{(1)}\mathcal{A}_N^{M-N} P \right].
\end{align}
In \eqref{predet},
we notice that the information of the particle configuration
$\{x_1, x_2,\dots,x_N \}$ is encoded in the determinant,
while the overall factor $K$ is independent of the configuration.
This fact allows us to determine the factor $K$ by evaluating
the overlap for a particular particle configuration. In fact, the
overlaps for some particular cases can be directly 
evaluated as in \cite{MSS,MSS2}.  For instance, we find
the following explicit expression for the case $x_j=j$ ($1\le j \le n$):
\begin{align}
\bra \psi(\{u\}_N)|12\dots N \ket
=\alpha^{N(N-1)/2} \prod_{j=1}^N
u_j^{N-1}(\alpha u_j-u_j^{-1})^{M-N} , \label{stepoverlap}
\end{align}
which can be evaluated with the help of its graphical description,
just in the same way with
the $n=0$ case in the intermediate scalar products \eqref{initial}.
Comparison of \eqref{stepoverlap} with \eqref{predet} for 
$x_j=j$ ($1\le j \le N$) determines
the desired prefactor $K$:
\begin{align}
K=\frac{\prod_{j=1}^N  (\alpha u_j-u_j^{-1})^M u_j^{2N-1}}
{\prod_{1 \le j<k \le 1}(u_j^2-u_k^2)}, \label{factor}
\end{align}
where we have used the Vandermonde determinant
$\mathrm{det}_N(x_j^{N-k})=\prod_{1 \le j<k \le N}(x_j-x_k)$
to evaluate the determinant in \eqref{predet} for the case
$x_j=j$ ($1\le j \le N$).
Insertion of the result of $K$ into  \eqref{predet} yields
\eqref{generaloverlap}.

We can also evaluate the dual expression \eqref{generaloverlaptwo} 
in the similar manner.
In this case the corresponding matrix product representation is given by
\begin{align}
\bra x_1\dots x_N |\psi(\{u\}_N)\ket=
\Tr_{W^{\otimes N}}\left[
\mathcal{A}_N^{M-x_N}
\mathcal{C}_N
\mathcal{A}_N^{x_N-x_{N-1}-1}
\dots\mathcal{C}_N\mathcal{A}_N^{x_2-x_1-1}\mathcal{C}_N\mathcal{A}_N^{x_1-1} Q\right],
\label{MP2}
\end{align}
where $\mathcal{C}_N$ is an element of the monodromy matrix defined in \eqref{decomp}
and $Q$ is a projection operator $Q=|1^N\ket\bra 0^N|$ acting on 
$W_1\otimes \dots\otimes W_N$ (cf. \eqref{overlap} and \eqref{reov}).  
The algebraic relations satisfied by the operators $\mathcal{A}$ and $\mathcal{C}$
are summarized in the following lemma.
\begin{lemma}\label{algebra2}
There exists a decomposition of $\mathcal{C}_n$:
$\mathcal{C}_n=\sum_{j=1}^n \mathcal{C}_n^{(j)}$ such 
that the following algebraic relations hold for $\mathcal{A}_n$ 
and $\mathcal{C}_n^{(j)}$:
\begin{align}
&\mathcal{A}_n \mathcal{C}_n^{(j)}
=\frac{u_j}{\alpha u_j-u_j^{-1}}\mathcal{C}_n^{(j)} \mathcal{A}_n , \\
&(\mathcal{C}_n^{(j)})^2=0, \\
&(\alpha u_k^2-1)\mathcal{C}_n^{(j)}\mathcal{C}_n^{(k)}
=-(\alpha u_j^2-1)\mathcal{C}_n^{(k)}\mathcal{C}_n^{(j)}, \ \ \ (j \neq k).
\end{align}
\end{lemma}
According to Lemma~\ref{algebra2} and the explicit expression for the wavefunction
\begin{align}
\bra 12\dots N |\psi(\{u\}_N)\ket=\alpha^{N(N-1)/2}\prod_{j=1}^{N} u_j^{M-1},
\end{align} the matrix product representation of the
overlap \eqref{MP2}
reduces to the determinant expression given in \eqref{generaloverlaptwo}.

\begin{example}
The wavefunction \eqref{generaloverlap}
for the configuration $x_j=2j-1$ ($1\le j \le N$) is obtained as follows.
\begin{align}
\bra \Omega|\prod_{j=1}^N C(u_j)|x_1 \cdots x_N \ket
&=\frac{\prod_{j=1}^N (\alpha u_j-u_j^{-1})^M u_j^{2N-1}}
{\prod_{1 \le j < k \le N}
(u_j^2-u_k^2)} \mathrm{det}_N\left[u_j^{-2k} (\alpha-u_j^{-2})^{-(2k-1)}\right] \nn \\
&=\frac{\prod_{j=1}^N (\alpha u_j-u_j^{-1})^{M-2N+1} u_j^{2N-2}}
{\prod_{1 \le j < k \le N} 
(u_j^2-u_k^2)} \mathrm{det}_N\left[ (\alpha u_j-u_j^{-1})^{2(N-k)}\right] \nn \\
&=\prod_{j=1}^N (\alpha u_j-u_j^{-1})^{M-2N+1}
\prod_{1\le j <k \le N} (\alpha^2 u_j^2 u_k^2-1).
\label{tasepoverlap}
\end{align}
From the second line to the third line we used the property of the
Vandermonde determinant.
The formula \eqref{tasepoverlap} for $\alpha=1$ and $M=2N$ recovers our 
former result \cite{MSS2} originally obtained by the Izergin-Korepin approach, i.e.,
deriving and solving 
recursive relations between different sizes of the overlap.
\end{example}

Finally let us show the following summation formulae for the
wavefunctions. 
\begin{theorem}
The off-shell wavefunction $\bra x_1\cdots x_N |\psi(\{v\}_N) 
\ket$ \eqref{generaloverlaptwo} satisfies the following summation formula:
\begin{align}
\sum_{1\le x_1\cdots \le x_N\le M}
\alpha^{M N-\sum_{j=1}^N x_j}
\bra x_1\cdots x_N | \psi(\{v\}_N) \ket
=\prod_{j=1}^N v_j^{M+1}\prod_{1\le j<k\le N}
\frac{1}{v_k^2-v_j^2} \mathrm{det}_N V,
\label{sum-wave1}
\end{align}
where $V$ is an $N\times N$ matrix with the elements
are 
\begin{align}
& V_{jk}=\sum_{m=0}^{j-1}(-1)^m \alpha^{M-m}
\binom{M}{m} v_k^{-2(m-j+1)} \quad (1\le j\le N-1), \nn \\
& V_{Nk}=-\sum_{m=\mathrm{max}(N-1,1)}^M (-1)^m \alpha^{M-m} \binom{M}{m}v_k^{-2(m-N+1)}.
\end{align}
While the dual off-shell wavefunction $\bra \psi(\{u\}_N)|x_1\cdots x_N\ket$
\eqref{generaloverlap} satisfies the following.
\begin{align}
\sum_{1\le x_1\cdots x_N \le N} \alpha^{\sum_{j=1}^N x_j-N}
\bra \psi(\{u\}_N) |x_1\cdots x_N\ket
=\prod_{j=1}^N u_j^{M+1}\prod_{1\le j<k \le N}\frac{1}{u_j^2-u_k^2}
\mathrm{det}_N\widetilde{V}
\label{sum-wave2}
\end{align}
with an $N\times N$ matrix $\widetilde{V}$ whose elements are given by 
\begin{align}
& \widetilde{V}_{jk}=\sum_{m=0}^{N-j}(-1)^m \alpha^{M-m}
\binom{M}{m} u_k^{-2(m+j-N)} \quad (2\le j\le N), \nn \\
& \widetilde{V}_{1k}=-\sum_{m=\mathrm{max}(N-1,1)}^M (-1)^m \alpha^{M-m} \binom{M}{m}u_k^{-2(m-N+1)}.
\label{det}
\end{align}
\end{theorem}
\begin{proof}
By the graphical description, it can be easily shown that
\begin{align}
&\bra \Omega | \lim_{u_j\to\infty} \prod_{j=1}^N u_j^{-M+1} C(u_j)
=\sum_{1\le x_1\cdots x_N\le M}\alpha^{MN-\sum_{j=1}^N x_j}
\bra x_1\cdots x_N|,  \nn \\
&\lim_{v_j\to\infty} \prod_{j=1}^N v_j^{-M+1} B(v_j) |\Omega\ket=
\sum_{1\le x_1\cdots x_N\le M} \alpha^{\sum_{j=1}^N x_j-N}
|x_1\cdots x_M\ket.
\end{align}
By substituting them into the determinant representation of the scalar product
\eqref{generalscalar}, one sees that the resultant expressions 
coincide with \eqref{sum-wave1} and \eqref{sum-wave2}. Here the 
limiting procedure $u_j\to \infty$ and $v_j\to \infty$ for $1\le j \le N$
in the scalar product \eqref{generalscalar} can be taken by expanding 
$d(u_j)=(u_j-\alpha u_j^{-1})^M$ and $d(v_k)=(v_k-\alpha v_k^{-1})^M$ in
the numerator of the elements for the determinant \eqref{element}, dividing
the numerator by the denominator and then making use of the formula:
\begin{align}
\lim_{u_j\to u}
\frac{\mathrm{det}_N[\Phi(u_j,v_k)]}{\prod_{1\le j < k \le N}(u_k-u_j)}
=\mathrm{det}_N\left[\frac{1}{(j-1)!}\left(\frac{\der}{\der u}\right)^{j-1}
\Phi(u,v_k)\right],
\label{limiting}
\end{align}
where $\Phi(u,v)$ is $(N-1)$-times differentiable functions of $u$.
\end{proof}
Setting $\alpha=1$, one finds that \eqref{sum-wave1} and \eqref{sum-wave2} 
recover the formula obtained in \cite{Bo}\footnote{Note that
there exists a misprint in the index of the summation of the
elements of the determinant corresponding to \eqref{det}}.
\section{Grothendieck polynomials and Cauchy identity}
%
The wavefunctions \eqref{generaloverlaptwo}  and \eqref{generaloverlap}
play a key role to analyze physically interesting
quantities such as Green functions.  Because the operators $B$
(or $C$) in \eqref{commutation} mutually commute, the wavefunctions 
(and the corresponding Green functions) can be described by some symmetric 
polynomials of $\{v\}_N$ (or $\{u\}_N$).
In this section, we show that the wavefunctions for generic value of $\alpha$ 
are written as Grothendieck polynomial which is a one-parameter 
deformation of  Schur polynomial.
Combining the completeness relation
and the determinant form of the scalar product \eqref{generalscalar},
one obtains the Cauchy identity of the Grothendieck polynomials.

Let us first define the Grothendieck polynomials.
\begin{definition}
The Grothendieck polynomial is defined to be the
following determinant \cite{IN}:
\begin{align}
G_\lambda(\bs{z};\beta)=
   \frac{\mathrm{det}_N(z_j^{\lambda_k+N-k}(1+\beta z_j)^{k-1})}
        {\prod_{1 \le j < k \le N}(z_j-z_k)},
 \label{GR}
\end{align}
where $\bs{z}=\{z_1,\dots,z_N\}$ is a set of variables and
 $\lambda$ denotes a Young diagram
$\lambda=(\lambda_1,\lambda_2,\dots,\lambda_N)$ with weakly decreasing
nonnegative integers $\lambda_1 \ge \lambda_2 \ge \dots \ge \lambda_N \ge 0$.
For our purpose, we further define the ``dual"  Grothendieck polynomial
(we discuss the orthogonality of the original and the  dual
Grothendieck polynomials later)
\begin{align}
\overline{G}_\lambda(\bs{z};\beta)=
\frac{\mathrm{det}_N(z_j^{\lambda_k+N-k}(1+\beta z_j^{-1})^{1-k})}
       {\prod_{1 \le j < k \le N}(z_j-z_k)}
\label{dualGR}.
\end{align}
\end{definition}
The Grothendieck polynomial \eqref{GR} and its dual version \eqref{dualGR}
can be regarded as a one-parameter deformation of the Schur polynomial, since 
they reduce to the Schur polynomial $s_\lambda(\bs{z})$ by taking
the parameter $\beta$ to be zero:
\begin{align}
G_\lambda(\bs{z};0)=\overline{G}_\lambda(\bs{z};0)=
s_\lambda(\bs{z}).
\label{Gr-Sch}
\end{align}
The Grothendieck polynomial was originally introduced in \cite{LS} as
polynomial representatives of Schubert classes in the
Grothendieck ring of the flag manifold.
From its origin, there are geometric studies \cite{Lascoux,Buch} related
to Schubert calculus, and also combinatorial ones
\cite{FK,LRS,BKSTY} as they are some classes
of symmetric polynomials.
However, it was shown very recently \cite{IN,IS} that
Grothendieck polynomials can be expressed in 
the determinant from \eqref{GR} (they moreover extended
the determinant representation to factorial Grothendieck polynomials
\cite{McNamara} originally defined in terms of
set-valued semi-standard tableaux). We take the determinant form
\eqref{GR} as the definition of the Grothendieck polynomials  in this paper.

Noticing that there exists one-to-one correspondence between the 
particle configuration $\{x_1,\dots,x_N\}$ ($1\le x_1 <\dots <x_N\le M$)
and the Young diagram 
$\lambda=(\lambda_1,\dots,\lambda_N) \subseteq (M-N)^N$
(which means $M-N\ge \lambda_1\ge \dots \ge \lambda_N\ge 0$), i.e. 
$\lambda_j=x_{N-j+1}-N+j-1$, one finds that the wavefunctions 
\eqref{generaloverlaptwo} and 
\eqref{generaloverlap} can be expressed as  Grothendieck polynomials \eqref{GR}.
\begin{lemma}\label{over-Gr}
By inserting the relation $\lambda_j=x_{N-j+1}-N+j-1$ and setting 
\begin{align}
z_j=\alpha-v_j^{-2}, \quad y_j^{-1}=\alpha-u_j^{-2}, \quad \beta=-1/\alpha,
\end{align}
the wavefunctions \eqref{generaloverlaptwo} and 
\eqref{generaloverlap} can, respectively, 
be expressed as the  Grothendieck polynomials \eqref{GR} and its dual 
version \eqref{dualGR}:
\begin{align}
&\bra x_1\dots x_N|\psi(\{v\}_N)\ket=\alpha ^{N(N-1)/2}\prod_{j=1}^N v_j^{M-1} 
G_\lambda(\bs{z};\beta), \nn \\
& \bra \psi(\{u\}_N)|x_1\dots x_N \ket=\alpha ^{N(N-1)/2}
\prod_{j=1}^N u_j^{M-1} y_j^{-M+N}(1+\beta y_j^{-1})^{N-1} 
\overline{G}_{\lambda} (\bs{y};\beta).
\label{wavefunction}
\end{align}
\end{lemma}

The Cauchy identity holding for the Schur polynomials can be extended to that for
the Grothendieck polynomials.
\begin{theorem} The following identity holds true for the Grothendieck polynomials
\eqref{GR} and \eqref{dualGR}.
\begin{align}
&\sum_{\lambda \subseteq
(M-N)^N}G_\lambda(\bs{z};\beta)
\overline{G}_\lambda(\bs{y};\beta)
\nonumber \\
&\qquad =\prod_{1 \le j < k \le N}\frac{1}{(z_j-z_k)(y_j-y_k)}
\mathrm{det}_N \left[ \frac{(z_j y_k)^{M}-
\left\{(1+\beta z_j)/(1+\beta y_k^{-1})\right\}^{N-1}}
{z_j y_k-1} \right].
\label{generalizedcauchy}
\end{align}
The usual Cauchy identity holding for the Schur polynomials is recovered
by taking $\beta=0$.
\end{theorem}
\begin{proof}
First, substituting the completeness relation, one decomposes 
the scalar product as
\begin{align}
\bra \psi(\{u\}_N)|\psi(\{v\}_N)\ket=\sum_{1\le x_1<\dots<x_N\le M}
\bra \psi(\{u\}_N)|x_1\dots x_N\ket\bra x_1\dots x_N |\psi(\{v\}_N)\ket.
\end{align}
Then substituting  the determinant representation for the 
scalar product \eqref{generalscalar} into the RHS of the
above and utilizing the relations in Lemma~\ref{over-Gr}
yields the one-parameter deformation of the Cauchy identity 
\eqref{generalizedcauchy}. 
\end{proof}
Taking $M\to \infty$, one has the following identity.
\begin{corollary} 
\begin{align}
\sum_{\lambda}G_\lambda(\bs{z};\beta)
\overline{G}_\lambda(\bs{y};\beta)=
\prod_{j=1}^N\left(\frac{1+\beta z_j}{1+\beta y_j^{-1}}\right)^{N-1}
\prod_{j,k=1}^N\frac{1}{1-z_j y_k},
\end{align}
where the sum is  over all Young diagram of shape 
$\lambda=(\lambda_1,\dots,\lambda_N)$. Taking $\beta=0$, the
well-known Cauchy identity for the Schur functions is recovered
\begin{align}
\sum_{\lambda}s_\lambda(\bs{z})
s_\lambda(\bs{y})=
\prod_{j,k=1}^N\frac{1}{1-z_j y_k}.
\end{align}
\end{corollary}

We also list the summation formulae for the Grothendieck
polynomials, which are obtained by inserting \eqref{wavefunction}
into \eqref{sum-wave1} and \eqref{sum-wave2}.
\begin{theorem}
The following summation formula is valid for the 
Grothendieck polynomials \eqref{GR}.
\begin{align}
&\sum_{\lambda\subseteq(M-N)^N}(-\beta)^{\sum_{j=1}^N\lambda_j}
G_{\lambda}(\bs{z};\beta)=\prod_{1\le j<k\le N}\frac{1}{z_k-z_j}
\mathrm{det}_N V^{(M)}
\label{sum-Gr}
\end{align}
with an $N\times N$ matrix $V^{(M)}$ whose matrix elements are
\begin{align}
&
V_{jk}^{(M)}=\sum_{m=0}^{j-1}(-1)^m (-\beta)^{j-N}
\binom{M}{m}(1+\beta z_k)^{m-j+N-1}
\quad (1\le j \le N-1), \nn \\
&
V_{Nk}^{(M)}=-\sum_{m=\mathrm{max}(N-1,1)}^M (-1)^m \binom{M}{m} (1+\beta z_k)^{m-1}.
\end{align}
While the dual Grothendieck polynomials \eqref{dualGR} satisfy
\begin{align}
\sum_{\lambda\subseteq(M-N)^N} (-\beta)^{-\sum_{j=1}^N \lambda_j}
\overline{G}_\lambda(\bs{y};\beta)
=\prod_{j=1}^N y_j^{M-1}\prod_{1\le j<k \le N}\frac{1}{y_k-y_j}
\mathrm{det}_N\widetilde{V}^{(M)},
\end{align}
where an $\widetilde{V}^{(M)}$ is an $N\times N$ matrix 
whose elements are given by 
\begin{align}
& \widetilde{V}_{jk}^{(M)}=\sum_{m=0}^{N-j}(-1)^m (-\beta)^{-j+1-M+N}
\binom{M}{m} (1+\beta y_k^{-1})^{m+j-N-1} \quad (2\le j\le N), \nn \\
& \widetilde{V}_{1k}^{(M)}=-\sum_{m=\mathrm{max}(N-1,1)}^M (-1)^m (-\beta)^{-M+N} 
\binom{M}{m}(1+\beta y_k^{-1})^{m-N}.
\end{align}
\end{theorem}
Finally, we discuss the orthogonality of the Grothendieck polynomials and
dual Grothendieck polynomials.
We now impose the periodic boundary condition
on the model, i.e.,
suppose that the spectral parameters $\{ z \}_N$ satisfy the
Bethe ansatz equations
\begin{align}
(1+\beta z_k)^N+(-1)^N z_k^M \prod_{j=1}^N (1+\beta z_j)=0
\quad (1\le k \le N).
\label{BAEintermsofz}
\end{align}
We insert into 
$\langle x_1 \cdots x_N|x_1^\prime \cdots x_N^\prime \rangle=
\prod_{j=1}^N \delta_{x_j x_j^\prime}$
the completeness of Bethe states
\begin{align}
I=\sum_{\{ u \}_N}
\frac{| \psi(\{ u \}_N) \rangle \langle \psi(\{ u \}_N)|}
{\langle \psi(\{ u \}_N)| \psi(\{ u \}_N) \rangle},
\end{align}
where the summation is over all of the solutions of the
Bethe ansatz equations.
We have
\begin{align}
\sum_{\{ u \}_N}
\frac{\langle x_1 \cdots x_N| \psi(\{ u \}_N)
\rangle \langle \psi(\{ u \}_N) |x_1^\prime \cdots x_N^\prime \rangle}
{\langle \psi(\{ u \}_N)| \psi(\{ u \}_N) \rangle}=\prod_{j=1}^N \delta_{x_j x_j^\prime},
\label{BAE-z}
\end{align}
which, with the use of the expressions
\eqref{norm} and \eqref{wavefunction},
can be translated to the following orthogonality relation
between the Grothendieck polynomials and
the  dual Grothendieck polynomials.
\begin{theorem}
The following orthogonality relation
between the Grothendieck polynomials and
the dual Grothendieck polynomials holds.
\begin{align}
\sum_{\{z \}_N}w(\{ z \}_N) 
\overline{G}_\lambda(\bs{z}^{-1};\beta)
G_\mu(\bs{z};\beta)=\delta_{\lambda \mu},
\label{orthogonal}
\end{align}
where the summation is over the all of the solutions
of the Bethe ansatz equation \eqref{BAEintermsofz},
and the weight $w(\{ z \}_N)$ given by
\begin{align}
w(\{ z \}_N)=\Bigg(1+\sum_{j=1}^N 
\frac{\beta z_j}{M+(M-N) \beta z_j}
 \Bigg)^{-1}
\prod_{\substack{j,k=1 \\ j \neq k }}^N (z_j-z_k)
\prod_{j=1}^N \frac{z_j^{1-N}(1+\beta z_j)}
{M+(M-N) \beta z_j}.
\end{align}
\end{theorem}
\begin{corollary}
Setting $\beta=0$ and taking the limit $M\to\infty$ yields
well-known orthogonal relation for the Schur polynomials 
(see \cite{Conrey} for example):
\begin{align}
\frac{1}{(2\pi {\mathrm{i}})^N N!}
\oint_{|z_1|=1} \cdots \oint_{|z_N|=1} 
\prod_{j=1}^N \mathrm{d} z_j
s_\lambda(\bs{z}^{-1}) s_\mu(\bs{z})
\prod_{\substack{j,k=1 \\ j \neq k}}^N (z_j-z_k)
\prod_{j=1}^{N} z_j^{-N} 
=\delta_{\lambda \mu}.
\label{Schur-orthogonal1}
\end{align}
\end{corollary}
\begin{proof}
Setting $\beta=0$ and using the relation \eqref{Gr-Sch},  one finds that 
\eqref{orthogonal} reduces to 
\begin{align}
\sum_{\{z \}_N}
s_\lambda(\bs{z}^{-1})
s_\mu(\bs{z})
\prod_{j=1}^N \frac{z_j}{M}
\prod_{\substack{j,k=1 \\ j \neq k}}^N (z_j-z_k)
\prod_{j=1}^{N} z_j^{-N}
=\delta_{\lambda \mu}.
\label{Schur-orthogonal2}
\end{align}
From the Bethe ansatz equation \eqref{BAE-z} for $\beta=0$, one observes that 
the roots  are located on the
unit circle in the complex plane: $z_j=\exp(2\pi \mathrm{i} I_j/M)$ where
$I_j\in \mathbb{Z}$ ($I_j\in (2\mathbb{Z}+1)/2$) for $N\in 2\mathbb{Z}+1$
($N\in 2\mathbb{Z}$) and $0\le I_1 < I_2<\cdots < I_N\le M-1$. Recalling the
sum in the above  is taken over all the sets of the solutions, and
ignoring the order of $\{z_j\}$, we can rewrite the sum as the multiple
integrals:
\begin{align}
\lim_{M\to\infty}\sum_{\{z\}_N} \prod_{j=1}^N \frac{z_j}{M}=
\lim_{M\to\infty}\sum_{\{z\}_N} \prod_{j=1}^N 
\frac{e^{2\pi \mathrm{i} I_j/M}}{M} =
\frac{1}{(2\pi \mathrm{i})^N N!}\oint_{|z_1|=1}\cdots
\oint_{|z_N|=1} \prod_{j=1}^N  \mathrm{d} z_j.
\end{align}
Inserting this limiting procedure into \eqref{Schur-orthogonal2}, 
one arrives at
\eqref{Schur-orthogonal1}.
\end{proof}

%
\section{Totally asymmetric simple exclusion process}\label{TASEP}
%
In the previous sections, we have evaluated the arbitrary off-shell
wavefunctions for the one-parameter family of the five vertex model
by making use of the matrix product representations. The most significant
is that the resultant determinant representation of the wavefunctions can
be expressed by  Grothendieck polynomials which is a one-parameter
deformation of  Schur polynomials. 
As mentioned in Section~\ref{FV}, the five vertex model includes 
several physically interesting models. As an application of the 
results obtained in the previous sections, we consider the TASEP 
and formulate the relaxation dynamics. 

The TASEP is a stochastic 
interacting particle system consisting of biased random walkers
obeying the exclusion principle, whose dynamics can be formulated 
as follows. We consider the $N$-particle system on the
periodic lattice with $M$ sites. By the exclusion rule, each site
can be occupied by at most one particle.
The dynamical rule of the TASEP  is: during the time interval 
$\mathrm{d} t$, a particle at a site $j$ jumps to the $(j+1)$th site 
with probability $\mathrm{d} t$, if the $(j+1)$th site is vacant.
The probability of being in the (normalized) state 
$|x_1 \cdots x_N \rangle$ is denoted as 
$P_t(x_1, \dots, x_N)$.
Then the arbitrary states can be written as
\begin{align}
|\varphi(t) \rangle=\sum_{1\le x_1<\dots<x_N\le M} 
P_t(x_1, \dots, x_N)| x_1 \cdots x_N \rangle.
\label{state}
\end{align}
Note that the probability is given as the
amplitude of each state, 
which is in contrast to the quantum mechanics where the 
probability is given by the squared magnitude of the 
amplitude. The time evolution of the state vector
is subject to the master equation
\begin{align}
\frac{\mathrm{d}}{\mathrm{d} t} |\varphi(t) \rangle
=\mathcal{H} |\varphi(t) \rangle. \label{master}
\end{align}
Here the stochastic matrix $\mathcal{H}$ of the TASEP is given by
\eqref{Baxter} for the case $\alpha=1$:
\begin{align}
\mathcal{H}=\sum_{j=1}^M \left\{
\sigma_j^+ \sigma_{j+1}^-
+\frac{1}{4}(\sigma_j^z \sigma_{j+1}^{z}-1)
 \right\}. \label{markov}
\end{align}
The eigenvalue spectrum of the 
stochastic matrix (\ref{markov})
can be calculated by the Bethe ansatz method \cite{GM,Bo,KBI,TF}
as formulated in Section~\ref{scalar}. 
Namely taking the logarithmic derivative of the 
eigenvalue of the transfer matrix \eqref{EV} according
to \eqref{Baxter}, and setting $w_j=1$ ($1\le j \le M$), 
$\alpha=1$ and $z_j=1-u_j^{-2}$  ($1\le j \le N$),
one obtains
\begin{align}
\mathcal{H}(\bs{z})=-N+\sum_{j=1}^N z_j^{-1},
\end{align}
where the parameters $\{ z \}_N$
must satisfy the Bethe ansatz equation \eqref{BAE}.
Explicitly it reads
\begin{align}
z_k^{-M} (1-z_k)^{N}=(-1)^{N-1}
\prod_{j=1}^N (1-z_j) \quad (1\le k \le N).
\label{BAE-tasep}
\end{align}
The state vector $| \psi(\bs{z})\ket$
(resp. $\bra  \psi(\bs{z})|$) defined by setting
$u_j^{-2}=1-z_j$ in $|\psi(\{u\}_N) \ket$
(resp. $\bra \psi(\{u\}_N)|$) 
becomes an energy eigenstate of \eqref{markov},
when we choose the set of parameters $\bs{z}$
as an arbitrary set of  solutions of \eqref{BAE-tasep}.
Then the norm of the eigenstate is given by \eqref{norm} after
setting $\alpha=1$ and $u_j^{-2}=1-z_j$.

The Green functions $\mathcal{G}_t(\bs{x}'|\bs{x})$ which is
the probability that the particles
starting at  initial positions $\bs{x}=\{x_1,\dots,x_N\}$
($1\le x_1<\cdots <x_N\le M$)
arrive at positions $\bs{x}'=\{x'_1,\dots,x'_N\}$ 
($1\le x'_1<\cdots <x'_N\le M$) at
time $t$ is given by solving the master equation \eqref{master}:
\begin{align}
\mathcal{G}_t(\bs{x}'|\bs{x})=
\bra x_1'\cdots x_N'|e^{\mathcal{H} t}|x_1\cdots x_N \ket.
\label{Green}
\end{align}
Utilizing the results in the previous section, one finds that
the Green function can be written in terms of the Grothendieck 
polynomials.
\begin{proposition}\label{Green-Grothendieck}
The Green function $\mathcal{G}_t(\bs{x}'|\bs{x})$ of the TASEP
whose stochastic matrix is given by \eqref{markov}
is expressed as the Grothendieck polynomials \eqref{GR} and \eqref{dualGR}
with $\beta=-1/\alpha=-1$:
\begin{align}
\mathcal{G}_t(\bs{x}'|\bs{x})=
\sum_{\bs{z}}
\frac{ 
      G_{\mu}(\bs{z};-1)\overline{G}_{\lambda}(\bs{z}^{-1};-1)}
{ \sum_{\gamma\subseteq (M-N)^N} 
      G_{\gamma}(\bs{z};-1) \overline{G}_{\gamma}(\bs{z}^{-1};-1)}
e^{\mathcal{H}(\bs{z}) t},
\label{Green-Gr}
\end{align}
where $\lambda=(\lambda_1,\dots,\lambda_N)$ and $\mu=(\mu_1,\dots,\mu_N)$
denote Young diagram characterized by the initial and
final positions: $\lambda_j=x_{N-j+1}-N+j-1$ and
$\mu_j=x'_{N-j+1}-N+j-1$, respectively. The arguments
of the Grothendieck polynomials
$\bs{z}=\{z_1,\dots,z_N\}$ and 
$\bs{z}^{-1}=\{z^{-1}_1,\dots,z^{-1}_N\}$ are expressed as  the
solutions to the Bethe ansatz equation \eqref{BAE-tasep}. 
The summation is over all the sets of the solutions to the  
Bethe ansatz equation.
\end{proposition}
\begin{proof}
Substituting the resolution of the identity operator
into \eqref{Green}, we have
\begin{align}
\mathcal{G}_t(\bs{x}'|\bs{x})=\sum_{\bs{z}}
\frac{\bra x_1'\cdots x_N'| \psi(\bs{z})\ket 
      \bra \psi(\bs{z})|x_1\cdots x_N\ket}
{\bra \psi(\bs{z}) | \psi(\bs{z}) \ket}e^{\mathcal{H}(\bs{z})t},
\label{resolution1}
\end{align}
where the parameters $\bs{z}=\{z_1,\dots,z_n\}$  are the solutions 
to the Bethe ansatz equation \eqref{BAE-tasep} and the summation is over
all the sets of the solutions. Finally utilizing the
expression of the wavefunctions \eqref{wavefunction}
and the deformed Cauchy identity \eqref{generalizedcauchy},
one arrives at \eqref{Green-Gr}.
\end{proof}

Let us check the validity of \eqref{resolution1} for the steady state.
After infinite time, the system will relax to
the steady state $|S_N\ket$:
\begin{align}
| S_N \rangle=
\binom{M}{N}^{-1}
 \sum_{1 \le x_1 < \dots < x_N \le M}
|x_1\cdots x_N\rangle.
\label{steady}
\end{align}
Up to some overall factor, the steady state corresponds to
the zero-energy 
state $|\psi(\{u\}_N)\ket$ with $u_j=\infty$ ($1\le j \le N$).
Due to the Perron-Frobenius theorem, all the
energy spectrum except for the unique zero eigenvalue must
have negative-real parts. Utilizing this fact and
substituting \eqref{steady} into \eqref{resolution1},
we have 
\begin{equation}
\mathcal{G}_\infty(\bs{x}'|\bs{x})=\binom{M}{N}^{-1}.
\label{Green-st}
\end{equation}
On the other hand, one finds that the Grothendieck polynomials
$G_{\lambda}(\bs{z},-1)$ and $\overline{G}_{\mu}(\bs{z}^{-1},-1)$
do not depend on the shapes $\lambda$ and $\mu$ in the limit $z_j\to 1$ ($1\le j \le N$):
\begin{align}
G_{\lambda}(\bs{z},-1)|_{\bs{z}\to\{1\}}=1, \quad
\prod_{j=1}^N (1-z_j)^{N-1}\overline{G}_{\mu}(\bs{z}^{-1},-1)|_{\bs{z} \to \{1\}}=1,
\end{align}
which follows from the formula \eqref{limiting}.
Thus the RHS of \eqref{Green-Gr}  reduces to $1/\sum_{\lambda\subseteq (M-N)^N} 1=
\binom{M}{N}^{-1}$ which is nothing but the RHS of \eqref{Green-st}. The following is a
consequence of Proposition~\ref{Green-Grothendieck} and the conservation law of the 
total probability: $\sum_{1\le x_1<\cdots<x_N\le M} P_t(x_1,\dots,x_N)=1$.
\begin{corollary}
The following sum rule holds for the Grothendieck polynomials.
\begin{align}
\sum_{\bs{z}}
\frac{ \sum_{\mu \subseteq (M-N)^N}
      G_{\mu}(\bs{z};-1)\overline{G}_{\lambda}(\bs{z}^{-1};-1)}
{ \sum_{\gamma\subseteq (M-N)^N} 
      G_{\gamma}(\bs{z};-1) \overline{G}_{\gamma}(\bs{z}^{-1};-1)}
e^{\mathcal{H}(\bs{z}) t}=1,
\end{align}
where the summation is over all the sets 
of the solutions to the  Bethe ansatz equation \eqref{BAE-tasep}.
\end{corollary}

Finally we comment on the relaxation dynamics of a physical quantity 
$\mathcal{A}$.  The time evolution of the expectation 
value for $\mathcal{A}$ starting from an initial 
state $| x_1 \cdots x_N \rangle$ is defined as
\begin{align}
\langle \mathcal{A} \rangle_t=\langle S_N| \mathcal{A} 
e^{\mathcal{M}t} | x_1 \cdots x_N  \rangle,
\label{expectation}
\end{align}
where $\langle S_N|$ is the left steady state vector
\begin{align}
\bra S_N |=\sum_{1 \le x_1 < \dots < x_N \le M}
\bra x_1 \cdots x_N |.
\end{align}
This definition comes from the fact that the TASEP is a stochastic process,
and the coefficient $P_t(x_1',\dots,x_N')$ of the state vector
$|\varphi(t) \rangle=e^{\mathcal{H}t} | x_1\cdots x_N \rangle$
directly gives the probability of being in the state
$|x_1' \cdots x_N' \rangle$ (see \eqref{state}),
and the left steady state vector $\langle S_N|$
plays the role of picking out the coefficients.
Inserting the resolution of identity as in \eqref{resolution1},
we can express the quantity (\ref{expectation}) in terms
of the Grothendieck polynomials.
\begin{proposition}
\begin{align}
\bra \mathcal{A}\ket_t=
\sum_{\bs{z}}
\frac{\left[\sum_{\nu\subseteq (M-N)^N}
\sum_{\mu \subseteq (M-N)^N} 
\mathcal{A}_{\mu}^{\nu} G_{\mu}(\bs{z};-1) \right] \overline{G}_{\lambda}(\bs{z}^{-1};-1)}
{ \sum_{\gamma\subseteq (M-N)^N} 
      G_{\gamma}(\bs{z};-1) \overline{G}_{\gamma}(\bs{z}^{-1};-1)}
e^{\mathcal{H}(\bs{z}) t},
\end{align}
where the matrix elements $\mathcal{A}_{\lambda}^{\mu}$ is given by
$\mathcal{A}_{\lambda}^{\mu}=\bra y_1\cdots y_N |\mathcal{A} |x_1\cdots x_N \ket$
with $x_j=\lambda_{N-j+1}+j$ ($M-N\ge \lambda_1\ge \cdots \lambda_N\ge 0$)
and  $y_j=\mu_{N-j+1}+j$  ($M-N\ge \mu_1\ge \cdots \mu_N\ge 0$),
and the summation is over all the sets of the solutions to the Bethe ansatz equation
\eqref{BAE-tasep}.
\end{proposition}
For instance,  the relaxation dynamics of the 
the local densities
$\mathcal{A}=n_i = 1-s_i$
and currents $\mathcal{A}= j_i =(1-s_i)s_{i+1}$
can be explicitly evaluated by applying the following theorem.
\begin{theorem}
Let $\mathcal{A}=s_l\cdots s_{l+n-1}$ ($-l+1\le n\le M$; $l\in \mathbb{Z}$). Then the
following formula holds for arbitrary complex values
$z_j\in\mathbb{C}$ ($1\le j \le N$):
\begin{align}
\sum_{\nu \subseteq (M-N)^N} \sum_{\mu \subseteq (M-N)^N}
 A_{\mu}^{\nu} G_{\mu}(z;-1)
=\prod_{j=1}^N z_j^{l+n-1} 
\prod_{1 \le j < k \le N}
\frac{1}{z_k-z_j} \mathrm{det}_N V^{(M-n)}, \label{formtwo} 
\end{align}
where $\mathcal{A}=s_l\cdots s_{l+n-1}$ and
the $N\times N$ matrix $V$ is written as
\begin{align}
& V_{jk}^{(M-n)}=\sum_{m=0}^{j-1} (-1)^m \frac{(M-n)!}{m!(M-m-n)!}
(1-z_k)^{m-j+N-1} \quad (1 \le j \le N-1), \nn \\
&
V_{Nk}^{(M-n)}=-\sum_{m=\mathrm{max}(N-1,1)}^{M-n} (-1)^m \frac{(M-n)!}{m!(M-m-n)!}
(1-z_k)^{m-1}.
\end{align}
\end{theorem}
\begin{proof}
The formula directly follows from the determinant representation
of the form factor for $\bra S_N |s_l\cdots s_{l+n-1}|\psi(\bs{z})\ket$
obtained in \cite{MSS,MSS2}.
\end{proof}

Setting $n=0$ and $l=1$ in the above formula, we find that
$\mathcal{A}_{\mu}^\nu=\delta^{\nu}_{\mu}$ and then
the above formula reduces to \eqref{sum-Gr} for $\beta=-1$.

\section{Conclusion}
In this paper, we studied the determinant structures of
a one-parameter family of integrable five vertex models.
By use of the algebraic Bethe ansatz and  the 
matrix product representation of the wavefunctions, 
the on/off-shell wavefunctions are expressed in terms 
of determinant forms. We found that the resultant 
wavefunctions are given by  Grothendieck polynomials 
which are a one parameter deformation of  Schur 
polynomials. By use of the properties satisfied by
the wavefunctions, we derived several important formulae
such as the Cauchy identity, summation formulae and so on 
for the Grothendieck polynomials. 

The Grothendieck polynomial was originally introduced in the 
context of Schubert calculus. This paper investigates the
objects of (geometric) representation theory
from the perspectives of integrable models.
See also \cite{KS,Ta,FWZ,Zu,ZJ,Korff,BBF,BMM} for
the integrable model approach to the
(geometric) representation theory or 
the classical integrable interpretation
of integrable models.
It is interesting to study the geometric and classical integrable
interpretation of the Cauchy identity, or to examine other
representation theoretical objects
from the integrable model side, the Littlewood-Richardson
coefficient for example.
The Cauchy identity also seems to have potential applications to
boxed plane partitions and determinantal process,
which we would like to pursue in the near future.

From the physics side, the evaluation of the wavefunctions by 
means of the matrix product representation allows us to formulate 
the exact relaxation dynamics of the periodic TASEP for arbitrary 
initial condition, beyond the step and alternating initial 
conditions studied in our former works \cite{MSS,MSS2}. 
We can now extract the asymptotics, fluctuations and so on 
from the formulation. Moreover, since we started from
the one-parameter extension of the $L$-operator
which corresponds to the TASEP with an effective long range potential 
\cite{DL,MP,SPS}, we are in a position to make an extensive study of 
them. One of the continuations of this paper is to study the properties
of the model.

\section*{Acknowledgement}
The authors thank C. Arita and K. Mallick for fruitful discussions
on the TASEP, especially the matrix product representation, and
T. Ikeda and H. Naruse on  Grothendieck polynomials, especially
informing us the determinant representation.
We also thank H. Katsura, H. Konno and M. Nakagawa for useful discussions.
The present work was partially supported
by Grants-in-Aid  for Scientific Research (C) No. 24540393
and for Young Scientists (B) No. 25800223.

\section*{Appendix  A}
Let us derive \eqref{loperator} as a solution to the $RLL$-relation \eqref{RLL}
with the $R$-matrix \eqref{Rmatrix},
by making the ansatz on the $L$ operator
\begin{align}
L_{\mu j}(u)=d_1(u)s_{\mu} s_j+d_2(u)\sigma_{\mu}^- \sigma_j^+
+d_3(u)\sigma_{\mu}^+ \sigma_j^-
+d_4(u)n_{\mu}s_j+d_5(u)n_{\mu}n_j, \label{ansatz}
\end{align}
where $d_j(u)$ are the functions to be determined.
In this paper, we consider the case $d_2(u)$ 
is not identically equal  to zero\footnote{
For $d_2(u)\equiv 0$, one sees from 
\eqref{rllsolve1}--\eqref{rllsolve8} that the model
reduces to the four vertex model: $d_1(u)=A u f(u)$,
$d_3(u)=f(u)$, $d_4(u)=d_5(u)=B u f(u)$, where
$A$ and $B$ are some constants, and $f(u)$ is
a rational function not identically equal to zero.} 
($d_2(u) \not\equiv 0$).
The equations to be solved are listed as
\begin{align}
&v d_1(u)d_2(v)-u d_1(v)d_2(u)=0, \label{rllsolve1} \\
&d_2(u)d_3(v)-d_2(v)d_3(u)=0, \label{rllsolve2} \\
&v d_5(u)d_2(v)-u d_5(v)d_2(u)=0, \label{rllsolve3} \\
&v d_1(u)d_3(v)-u d_1(v)d_3(u) =0, \label{rllsolve4} \\
&v d_5(u)  d_3(v)-u d_5(v) d_3(u)=0, \label{rllsolve5} \\
&(u^2-v^2)d_2(u)d_3(v)+uv(d_1(u)d_4(v)-d_1(v)d_4(u))=0, \label{rllsolve6} \\
&(u^2-v^2)d_5(u)d_2(v)+u(v d_2(u) d_4(v)- u d_2(v)d_4(u))=0, \label{rllsolve7} \\
&(u^2-v^2)d_5(u)d_3(v)+u(v d_3(u)d_4(v)-u d_3(v) d_4(u))=0. \label{rllsolve8}
\end{align}
From \eqref{rllsolve1}, \eqref{rllsolve2} and \eqref{rllsolve3},
we have the relations between $d_1(u)$, $d_3(u)$, $d_5(u)$ and $d_2(u)$
with the use of arbitrary constants $A$, $B$ and $C$ as
\begin{align}
d_1(u)&=Aud_2(u), \\
d_3(u)&=Bd_2(u), \\
d_5(u)&=Cu d_2(u).
\end{align}
Substituting the above relations into the remaining equations,
we find
\eqref{rllsolve4} and \eqref{rllsolve5} are automatically satisfied,
and we are left with \eqref{rllsolve6}, \eqref{rllsolve7} and \eqref{rllsolve8}
which now read
\begin{align}
&(Bd_2(u)+Aud_4(u))v^2 d_2(v)-(Bd_2(v)+Avd_4(v))u^2 d_2(u)=0, \\
&u(Cud_2(u)-d_4(u))d_2(v)-v(Cvd_2(v)-d_4(v))d_2(u)=0.
\end{align}
These equations lead to following relations between $d_4(u)$ and $d_2(u)$
\begin{align}
Bd_2(u)+Aud_4(u)&=Eu^2 d_2(u), \label{rllsolve9} \\
Cu^2 d_2(u)-u d_4(u)&=F d_2(u), \label{rllsolve10}
\end{align}
with constants $E$ and $F$.
Assuming $A \neq 0$, the compatibility between the two relations
\eqref{rllsolve9} and \eqref{rllsolve10}
leads to $E=AC, F=A^{-1}B$. Considering also the case $A=0$,
we finally find the elements of the $L$-operator
satisfying the $RLL$-relation under the ansatz
\eqref{ansatz} to be
\begin{align}
d_1(u)&=Auf(u), \\
d_2(u)&=f(u), \\
d_3(u)&=Bf(u), \\
d_4(u)&=(Cu-Du^{-1})f(u), \\
d_5(u)&=Cu f(u),
\end{align}
where $f(u)\not\equiv 0$ is a rational function of $u$
and $A,B,C,D$ are constants satisfying the constraints
$B-AD=0$.
Taking $A=B=D=1$ and $C=\alpha$, we have the desired 
$L$-operator \eqref{loperator} up to the overall factor 
$f(u)$. (See \cite{PM} for example for a brute force search of
more complicated integrable models of higher ranks
or higher spins.)

\end{document}